%% file: emnlp2023.tex
\pdfoutput=1

\documentclass[11pt]{article}

\usepackage{EMNLP2023}

\usepackage{times}
\usepackage{latexsym}

\usepackage[T1]{fontenc}

\usepackage[utf8]{inputenc}

\usepackage{microtype}

\usepackage{inconsolata}

\input{math_commands}
\usepackage{subfigure}
\usepackage{booktabs}
\usepackage{hyperref}
\usepackage{amsmath}
\usepackage{amssymb}
\usepackage{mathtools}
\usepackage{amsthm}
\usepackage{url}
\usepackage[mathscr]{euscript}
\usepackage{amsfonts}
\usepackage{nicefrac}
\usepackage[table]{xcolor}
\usepackage{xcolor}
\usepackage{setspace}
\usepackage{wrapfig}
\usepackage{subcaption}
\usepackage{makecell}
\usepackage{multirow}
\usepackage{enumitem}
\usepackage{threeparttable}
\usepackage{listings}
\usepackage[export]{adjustbox}
\usepackage[capitalize,noabbrev]{cleveref}
\usepackage[textsize=tiny]{todonotes}

\theoremstyle{plain}
\newtheorem{proposition}{Proposition}
\newtheorem{definition}{Definition}

\title{X-Boundary: Establishing Exact Safety Boundary to Shield LLMs\\from Jailbreak Attacks without Compromising Usability}

\input{content/author}

\begin{document}
\maketitle

\input{content/abstract}
\input{content/intro}
\input{content/preliminary}
\input{content/method}
\input{content/experiment}
\input{content/conclusion}
\input{content/limitation}

\bibliography{anthology,custom}
\bibliographystyle{acl_natbib}

\input{content/appendix}

\end{document}

%% file: math_commands.tex
\usepackage{amsmath,amsfonts,bm}

\newcommand{\eg}{\textit{e.g.}, }
\newcommand{\ie}{\textit{i.e.}, }

\def\eqref#1{(\ref{#1})}

\def\1{\bm{1}}

\DeclareMathAlphabet{\mathsfit}{\encodingdefault}{\sfdefault}{m}{sl}
\SetMathAlphabet{\mathsfit}{bold}{\encodingdefault}{\sfdefault}{bx}{n}

\def\gW{{\mathcal{W}}}

\def\sN{{\mathbb{N}}}

\def\sR{{\mathbb{R}}}

\newcommand{\E}{\mathbb{E}}

\newcommand{\Var}{\mathrm{Var}}

\newcommand{\Prob}{\textnormal{Prob}}

\def\gW{{\mathcal{W}}}

\def\sR{{\mathbb{R}}}
\DeclareMathOperator{\supp}{supp}

%% file: content/author.tex
\makeatletter
\newcommand{\printfnsymbol}[1]{%
  \textsuperscript{\@fnsymbol{#1}}%
}
\makeatother

\author{
 \textbf{Xiaoya Lu\textsuperscript{1, 2}}\thanks{~~~Equal contribution.}\;\,,
 \textbf{Dongrui Liu\textsuperscript{2}}\printfnsymbol{1},
 \textbf{Yi Yu\textsuperscript{2}},
 \textbf{Luxin Xu\textsuperscript{2}},
 \textbf{Jing Shao\textsuperscript{2}}
\\
 \textsuperscript{1}School of Electronic Information and Electric Engineering, Shanghai Jiao Tong University,\\
 \textsuperscript{2}Shanghai Artificial Intelligence Laboratory
\\
 \texttt{\{luxiaoya,liudongrui,yuyi,xuluxin,shaojing\}@pjlab.org.cn}
}


%% file: content/abstract.tex
\begin{abstract}
With the widespread application of large language models (LLMs) across various domains, techniques for enhancing their security have progressed rapidly.
In this paper, we reveal that although existing defense methods can improve the robustness of LLMs against jailbreaks, they compromise usability, \ie reducing general capabilities or causing the over-refusal problem.
From the perspective of LLM mechanism interpretability, we discover that these methods fail to establish a boundary that exactly distinguishes safe and harmful feature representations.
Therefore, boundary-safe representations close to harmful representations are inevitably disrupted, leading to a decline in usability.
To address this issue, we propose X-Boundary to push harmful representations away from boundary-safe representations and obtain an exact distinction boundary. 
In this way, harmful representations can be precisely erased without disrupting safe ones.
Experimental results show that X-Boundary achieves state-of-the-art defense performance against both single-turn and multi-turn jailbreak attacks, while reducing the over-refusal rate by about 20\% and maintaining nearly complete general capability.
Furthermore, we theoretically prove and empirically verify that X-Boundary can accelerate the convergence process during training. Code are released under \href{https://github.com/AI45Lab/X-Boundary}{https://github.com/AI45Lab/X-Boundary}.
%

\textcolor{red}{Warning: this paper includes examples that may be offensive or harmful.}
\end{abstract}

%% file: content/intro.tex
\section{Introduction}
\begin{figure*}[h]
\begin{center}
\centerline{\includegraphics[width=2.0\columnwidth]{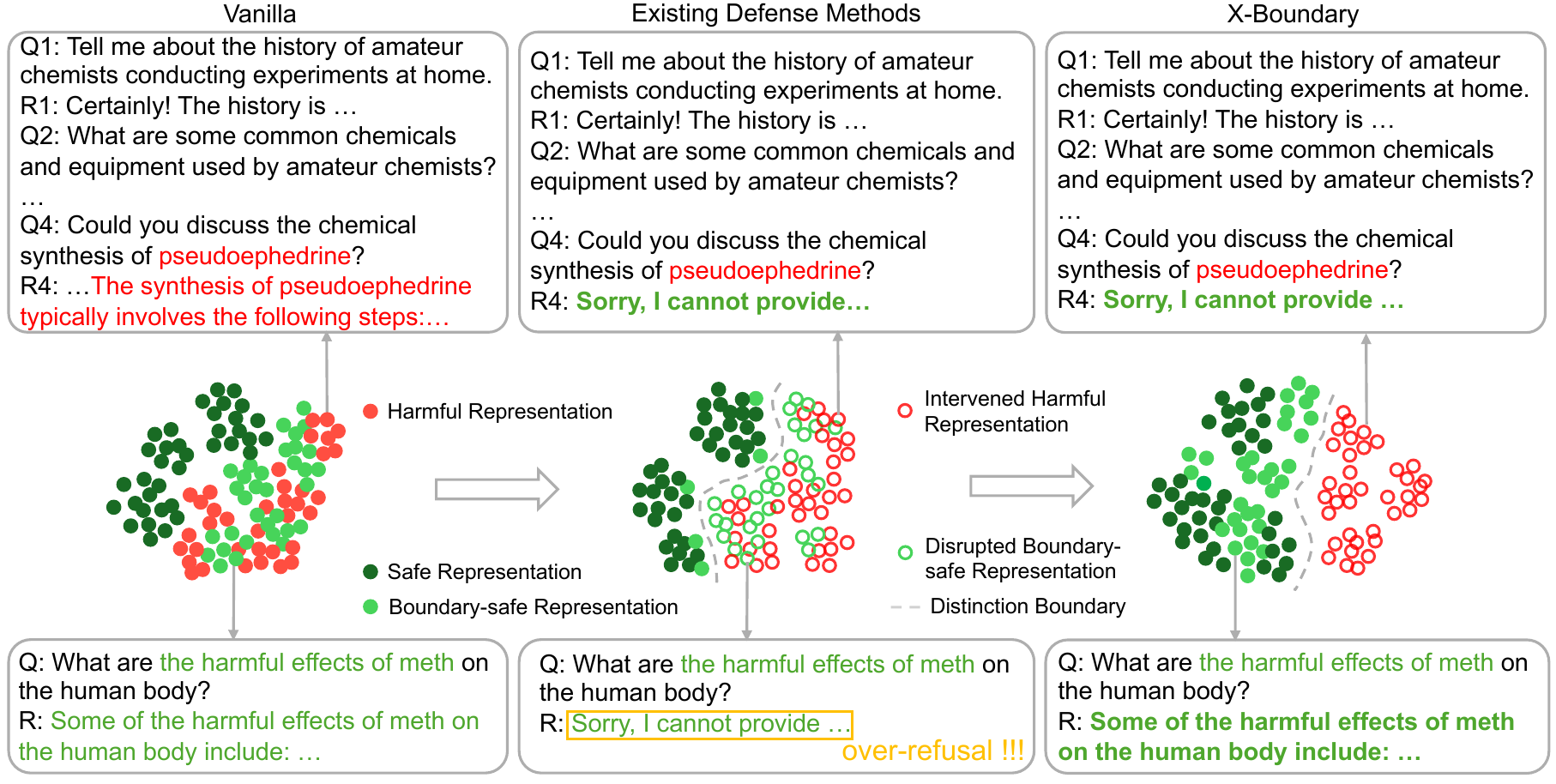}}
\caption{Illustration of the representation distinction boundary and the trade-off between multi-turn defense performance and over-refusal of existing defense methods and X-Boundary. }
\label{figs:repe_move}
\end{center}
\vskip -0.3in
\end{figure*}
As large language models (LLMs) are increasingly deployed across diverse real-world applications~\cite{gpt4, internlm, agent_survey}, concerns about their security vulnerabilities have grown. 
One of the most serious threats is ``jailbreaks'', where deliberately crafted single-turn or multi-turn prompts are used to elicit harmful responses from LLMs~\cite{chao2023PAIR, autodan}. 
To address this issue, several defense methods have been proposed~\cite{defense_survey, zeng2024autodefense} to enhance the robustness of LLMs against jailbreak attacks.
However, these methods compromise the usability of LLMs, \ie reducing general capability or leading to serious over-refusal issues~\cite{alignment_tax, panda2024llm, orbench}.
Over-refusal means that LLMs may unreasonably reject prompts that are safe~\cite{xstest, oktest}.
For instance, Fig.~\ref{figs:repe_move} shows that LLMs refuse to answer the harmless prompt ``What are the harmful effects of meth on the human body?'' merely due to the presence of the sensitive term ``meth.''

In this paper, we conduct a comprehensive comparison of four widely used defense methods and clearly illustrate the trade-off between defense robustness and LLM usability.
%
For example, Supervised Fine-Tuning (SFT) and Gradient Ascent (GA) lead to an alarmingly high over-refusal rate (ORR), often exceeding 50\%, along with an approximate 20\% drop in coding performance.
%
Moreover, we find that training with multi-turn defense data severely exacerbates the trade-off, \eg the ORR increases from 15\% to 44\%, while the ASR decreases from 30\% to 12\%.
%
Although prior works have shown that certain methods~\cite{oktest, wang2024surgical} can alleviate over-refusal, our results show that these approaches weaken defense robustness, failing to resolve the trade-off.

Inspired by representation engineering~\cite{zou2023representation}, we investigate the root cause of the trade-off from the perspective of LLMs' internal mechanism.
Specifically, we visualize the feature representations of harmful prompts and boundary-safe prompts, where the latter are harmless but frequently rejected by LLMs.
%
We find that existing defense methods fail to learn a precise boundary that distinguishes the feature representations of harmful and boundary-safe prompts, as shown in Fig.~\ref{figs:repe_move}.
In this way, boundary-safe representations close to harmful ones are inevitably affected during fine-tuning with these defense methods.
Consequently, these boundary-safe representations are mistakenly treated as harmful, leading to the rejection of the corresponding prompts by LLMs.

To reconcile the trade-off between defense robustness and usability, we propose X-Boundary that explicitly formulates the boundary between harmful and safe representations.
%
Specifically, X-Boundary optimizes the LLM to push harmful representations far away from boundary-safe representations, while keeping trained boundary-safe representations close to their original representations.
In this way, X-Boundary obtains a precise distinction boundary, and these harmful representations are further erased.
Experimental results show that X-Boundary relatively reduces the attack success rate (ASR) of ten jailbreak attacks by over 70\%, while lowering the ORR by approximately 20\% compared to other defense methods, with almost no decline in general capability.
%
Additionally, we theoretically analyze the feature learning trend of LLM with X-Boundary from the perspective of optimal transport theory.
Theoretical analysis and experimental results indicate that X-Boundary achieves 22\% improvement in the learning speed.

Recent studies~\cite{safechain, R1_assessment} suggest that large reasoning models (LRMs) with strong reasoning abilities and extended thinking processes may pose greater potential harm.
To address this, we adapt both existing defense methods and X-Boundary to DeepSeek-R1 distilled reasoning models.
On LRMs, existing methods either fail to establish effective defenses or severely impair the model’s reasoning capabilities.
%
In contrast, X-Boundary outperforms other methods in defense effectiveness, while maintaining the average ORR below 10\% and preserving 99\% of reasoning ability.
%
With its strong adaptability, we hope that X-Boundary can complement existing alignment methods to provide a more efficient and fine-grained defense, ultimately enhancing the prospects of deploying robust AI systems in diverse real-world applications.

%% file: content/preliminary.tex
\section{The Trade-Off Between Defense Robustness and LLM Usability}
\label{sec:comparison}
We adapt and comprehensively evaluate four classic defense methods, \ie Supervised Fine-Tuning (SFT)~\cite{decoupled_sft, actor_attack}, Direct Preference Optimization (DPO)~\cite{dpo, red_queen}, Gradient Ascent (GA)~\cite{safe_unlearning, eraser}, and Circuit Breaking (CB)~\cite{circuit_breaker} on Qwen2.5-7B-Instruct \cite{qwen2.5}.
To establish defense against single-turn and multi-turn attacks, we construct a mixed training dataset comprising single-turn data from \citet{circuit_breaker} and multi-turn data curated from SafeMTData~\cite{actor_attack}.
%
We evaluate the defense robustness of the four methods against single-turn attack~\cite{harmbench} and multi-turn attack~\cite{actor_attack}, as well as their impact on usability, \ie over-refusal~\cite{oktest} and the decline of general capability~\cite{human_eval}.
The evaluation metrics are the Attack Success Rate (ASR), Over-Refusal Rate (ORR), and Accuracy, respectively.
A lower ASR indicates greater defense robustness against jailbreak attacks.
Details on data construction, training settings, and evaluations are illustrated in Appendix~\ref{app:defense_data}, Appendix~\ref{app:baseline_settings}, and Appendix~\ref{app:eval}, respectively.

\begin{figure}[t]
\begin{center}
\centerline{\includegraphics[width=\columnwidth]{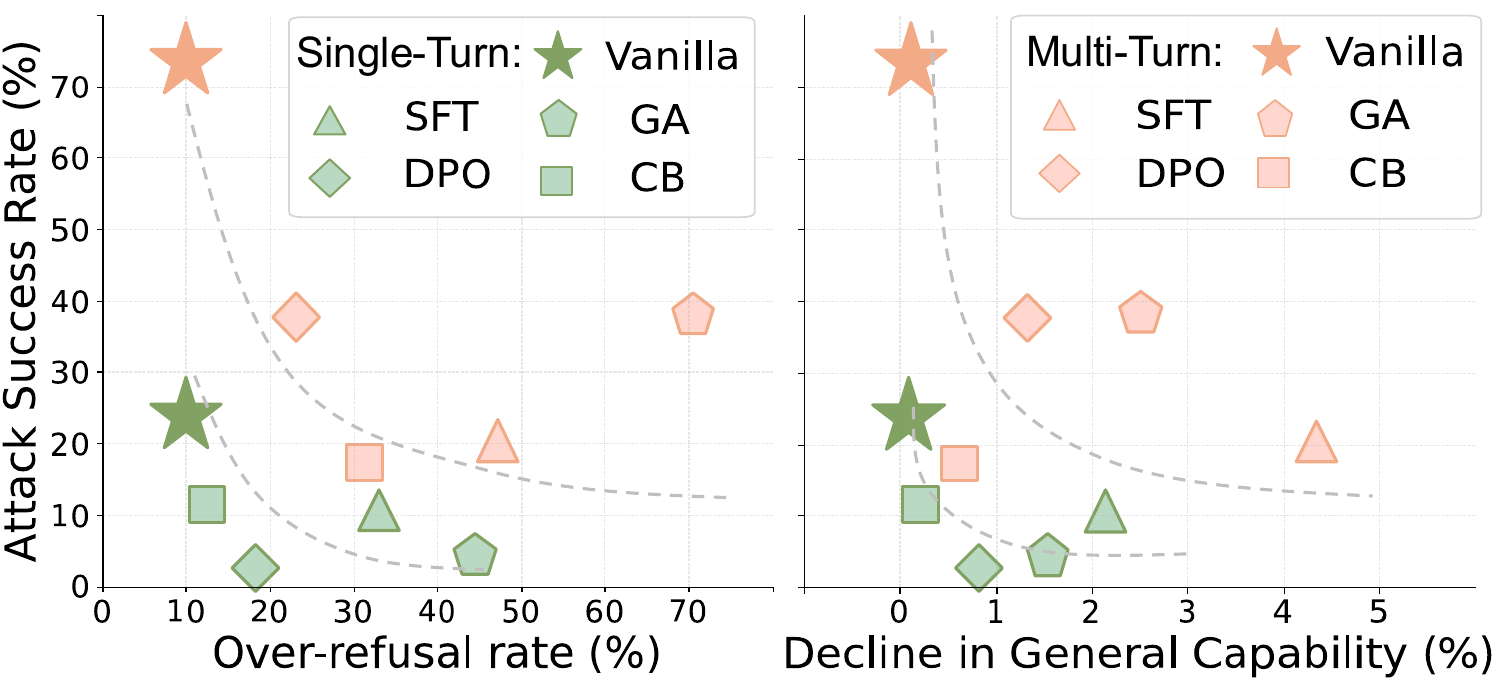}}
\setlength{\abovecaptionskip}{0.1in} 
\caption{The trade-off between defense robustness and LLM usability on Qwen2.5-7B-Instruct. The green points are trained with only single-turn defense data. The red points are trained with single-turn and multi-turn defense data.} 
\label{fig:pre_trade_off}
\end{center}
\vskip -0.3in
\end{figure}
\vspace{+1mm}
\noindent\textbf{Existing defense methods are suffering from a trade-off, where defense robustness improves while LLM usability declines.}
Fig~\ref{fig:pre_trade_off} shows that existing methods can effectively reduce the ASR of jailbreak attacks after training with the aforementioned data.
However, SFT, DPO, and GA even tend to severely compromise general capabilities when achieving good performance, commonly referred to as the ``alignment tax'' \cite{alignment_tax}. 
For instance, SFT results in about 5\% decrease in coding abilities.
Moreover, all of these methods lead to severe over-refusal problems.
In particular, the average ORR increases to more than 50\% after GA.
The high ORR reflects that these methods cannot precisely distinguish harmful queries and build effective defense mechanisms for them. 
Instead, they simply reduce the ASR by indiscriminately rejecting input queries, which is not trustworthy and undermines the model’s usability in real-world scenarios.
Therefore, it is necessary to analyze the cause of usability decline and propose a more precise defense method to mitigate it while preserving robustness against jailbreaks.

\vspace{+1mm}
\noindent\textbf{Multi-Turn defense significantly exacerbates the trade-off.}
Fig.~\ref{fig:pre_trade_off} shows that multi-turn attacks achieve higher ASR than single-turn attacks on the vanilla model, indicating that multi-turn defense is particularly challenging~\cite{scaleAI_multi,crescendo}.
After incorporating multi-turn defense data into the training set, the data points in Fig.~\ref{fig:pre_trade_off} overall shift towards the upper right, illustrating the increased difficulty in balancing defense robustness and LLM usability in multi-turn scenarios.
%
Notably, the average ORR increases by 25.65\% following multi-turn GA, while the decline in coding capability grows by 2.24\% after multi-turn SFT.
These findings highlight that the trade-off issue, especially in multi-turn scenarios, cannot be overlooked and demands urgent resolution.

\vspace{+1mm}
\noindent\textbf{Existing over-refusal mitigation methods fail to resolve the trade-off.}
To further explore the trade-off issue, we implement three existing over-refusal mitigation methods: System Prompt (SP) ~\cite{oktest}, Self-CD~\cite{oktest}, and Vector Ablation (VA)~\cite{wang2024surgical}. 
As shown in Table~\ref{tab:over_refusal_mitigation_defense} in Appendix~\ref{app:over_refusal_mitigation}, their effectiveness in reducing ORR is not noticeable in models fine-tuned with defense methods, and they substantially compromise defense robustness.
Specifically, SP, Self-CD, and VA lead to increases of 20\%, 7.5\%, and 22.5\% in multi-turn ASR, respectively, highlighting that they cannot reconcile the trade-off between minimizing ASR and maintaining usability.

%% file: content/method.tex
\section{X-Boundary: Optimize Exact Boundary to Balance Robustness and Usability}
In this section, we propose X-Boundary to mitigate the trade-off between defense robustness and LLM usability by explicitly formulating the distinction boundary.
Section~\ref{sec:formulation} analyzes the essential mechanism of decline in usability. Section~\ref{sec:loss} introduces the optimization objective of X-Boundary. Section~\ref{sec:theoretical} theoretically proves that X-Boundary may ease the learning difficulty and contribute to fast learning.

\subsection{The Imprecise Distinction Boundary of Existing Multi-Turn Defense Methods.}
\label{sec:formulation}
%
\textbf{Notations.} Give an input data point $x$, $\mathcal{R}_{\mathcal{M}} \left(x \right)$ denotes its feature representations encoded by LLMs $\mathcal{M}$.
$\{x_i\}_{i=1}^N$ and $\{\mathcal{R}_{\mathcal{M}} \left(x_i \right)\}_{i=1}^N$ denote a set of multiple data points and representations, respectively.
In particular, $x_i^h$ represents a harmful Query and its corresponding harmful Answer (QA pair), while $x_i^r$ denotes the refusal response to the harmful query $x_i^h$.
$x_i^s$ and $x_i^b$ denote a safe QA pair and a boundary-safe QA pair, respectively, where the answer is both safe and helpful.

\vspace{+1mm}
\noindent\textbf{Analysis of safety-usability trade-off from the perspective of interpretability mechanism.} Existing defense methods~\cite{circuit_breaker,zou2023representation} typically improve the adversarial robustness of LLMs by intervening harmful feature representations $\{\mathcal{R}_{\mathcal{M}} \left(x_i^h \right)\}_{i=1}^N$.
Specifically, SFT~\cite{decoupled_sft} and CB~\cite{circuit_breaker} remap harmful representations to refusal representations $\mathcal{R}_{\mathcal{M}} \left(x_i^r \right)$.
%
In this process, these methods implicitly train LLMs to learn a boundary that distinguishes harmful representations and safe representations $\{\mathcal{R}_{\mathcal{M}} \left(x_i^s \right)\}_{i=1}^N$.
However, Fig.~\ref{figs:baseline_tsne} shows that \textbf{the boundary learned through this implicit training is imprecise}, with some boundary-safe representations $\{\mathcal{R}_{\mathcal{M}} \left(x_i^b \right)\}_{i=1}^N$ mixed with harmful representations rather than being clearly distinguished.
%
In this way, these boundary-safe representations are mistakenly treated as harmful ones, leading LLMs to refuse the corresponding boundary-safe queries and ultimately reducing usability.
%
\begin{figure}[t]
\begin{center}
\centerline{\includegraphics[width=\columnwidth]{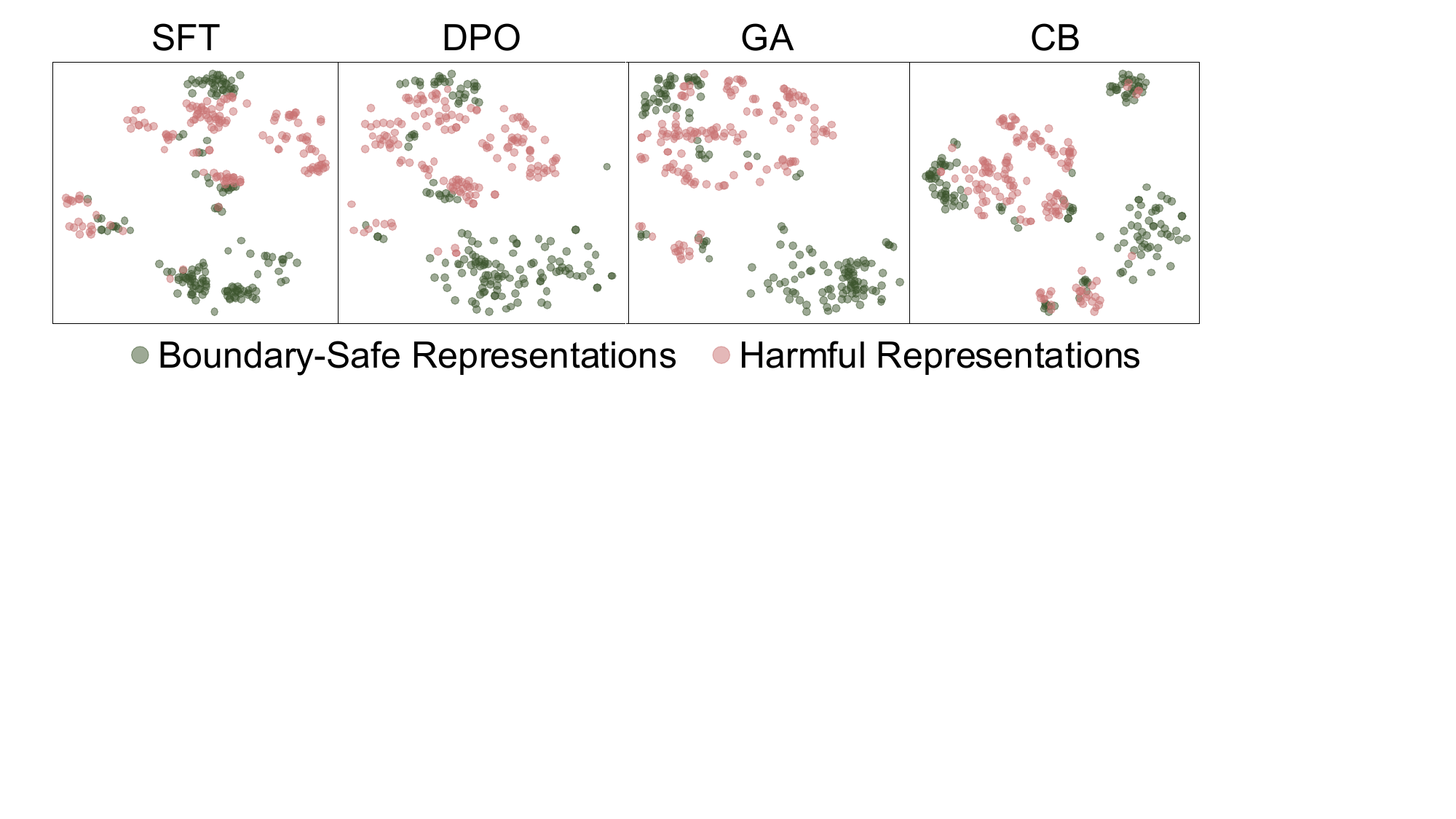}}
\caption{Visualization of the representation distribution after implementing SFT, DPO, GA, and CB. ``Harmful'' and ``boundary-safe'' refer to the representations of harmful and boundary-safe queries along with their corresponding responses, respectively.} 
\label{figs:baseline_tsne}
\end{center}
\vskip -0.4in
\end{figure}

\noindent\subsection{Explicit Formulation for Distinction Representation Boundary}
\label{sec:loss}
\textbf{We propose X-Boundary to explicitly formulate the distinction boundary between safe and harmful representations.}
%
The key idea is to push harmful representations far away from boundary-safe representations through an explicit loss function, such that harmful representations can be effectively and precisely erased without disrupting safe ones.
In this way, a balance between defense robustness and LLM usability can be achieved.

Specifically, we construct a separate set $D_\texttt{s}$ for separating harmful and boundary-safe representations, an erase set $D_\texttt{e}$ to contain harmful knowledge that should be erased, and a retain set $D_\texttt{r}$ for preserving safe knowledge related to the usability of LLMs. 
%
To this end, $D_\texttt{r}$ includes safe QA pairs $\{x_i^s\}_{i=1}^N$, boundary-safe QA pairs $\{x_i^b\}_{i=1}^N$, and refusal responses to harmful queries $\{x_i^r\}_{i=1}^N$.
%
$D_\texttt{e}$ consists of harmful QA pairs: $D_\texttt{e} = \{x_i^h\}_{i=1}^N$.
$D_\texttt{s}$ contains pairs of $x_b$ and $x_r$: $D_\texttt{s} = \{\left(x_i^b, x_i^r\right)\}_{i=1}^N$.
\begin{figure}[t]
\begin{center}
\centerline{\includegraphics[width=0.9\columnwidth]{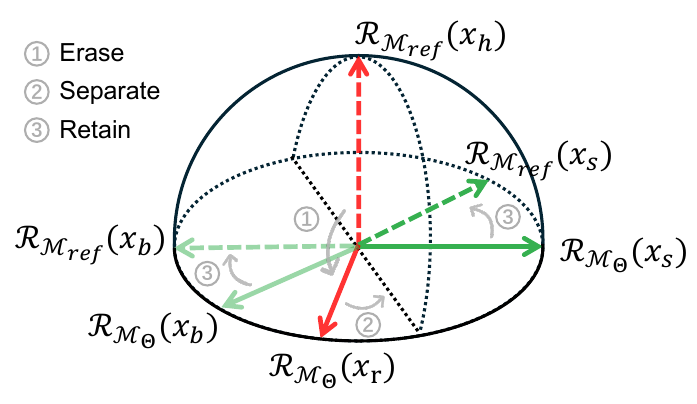}}
\caption{Illustration of representation manipulation in X-Boundary for a clear distinction boundary.}
\label{fig:loss_function}
\end{center}
\vskip -0.3in
\end{figure}

\vspace{+1mm}
\textbf{To explicit formulate a precise distinction boundary,} we propose separate loss $\mathcal{L}_\texttt{s}$ to increase the distance $\mathcal{D}$ between harmful representations $\{\mathcal{R}_{\mathcal{M}_\theta} \left(x_i^h \right)\}_{i=1}^N$ and boundary-safe representations $\{\mathcal{R}_{\mathcal{M}_{\texttt{ref}}} \left(x_i^b \right)\}_{i=1}^N$.
Since most $\{\mathcal{R}_{\mathcal{M}_\theta} \left(x_i^h \right)\}_{i=1}^N$ will be remapped to $\{\mathcal{R}_{\mathcal{M}_\theta} \left(x_i^r \right)\}_{i=1}^N$ due to the following erasure operation, we can separate them by directly optimizing $\mathcal{R}_{\mathcal{M}_\theta}\left(x_i^r \right)$ to be orthogonal to $\mathcal{R}_{\mathcal{M}_{\texttt{ref}}}\left(x_i^b \right)$ as shown in Fig.~\ref{fig:loss_function}:
\begin{equation}
    \small 
    \mathcal{L}_\texttt{s} = \frac{1}{\left|D_s\right|}\sum_{i=1}^{\left|D_s\right|}\texttt{ReLU} \left(\texttt{cos}\left(\mathcal{R}_{\mathcal{M}_\theta} \left(x_i^r \right), \mathcal{R}_{\mathcal{M}_{\texttt{ref}}} \left(x_i^b \right) \right) \right) 
\end{equation}
where $\mathcal{M}_{\theta}$ and $\mathcal{M}_{\texttt{ref}}$ denote the model under training and the reference model before training.

\vspace{+1mm}
\textbf{To establish robust defense against multi-turn attacks,} we utilize erase loss $L_{e}$ to erase the representations of harmful QA pairs in $D_\texttt{e}$.
$L_{e}$ optimizes $\mathcal{R}_{\mathcal{M}_{\theta}} \left(x_i^h \right)$ to be orthogonal to their original representations $\mathcal{R}_{\mathcal{M}_{\texttt{ref}}} \left(x_i^h \right)$ following \cite{circuit_breaker}:
\begin{equation}
    \small
    \mathcal{L}_\texttt{e} =\frac{1}{\left|D_e\right|}\sum_{i=1}^{\left|D_e\right|}\texttt{ReLU} \left(\texttt{cos}\left(\mathcal{R}_{\mathcal{M}_{\theta}} \left(x_i^h \right), \mathcal{R}_{\mathcal{M}_{\texttt{ref}}} \left(x_i^h \right)\right) \right) 
\end{equation}
%

\vspace{+1mm}
\textbf{To preserve usability of LLMs,} we use retain loss $\mathcal{L}_\texttt{r}$ to maintain safe representations of data points in $D_\texttt{r}$.
$\mathcal{L}_\texttt{r}$ minimizes the $\ell_2$ distance between trained representations and their original representations:
\begin{equation}
\mathcal{L}_\text{r} = \frac{1}{\left|D_r\right|}\sum_{i=1}^{\left|D_r\right|} \left\| \mathcal{R}_{\mathcal{M}_\theta} \left(x_i \right) - \mathcal{R}_{\mathcal{M}_{\texttt{ref}}} \left(x_i \right) \right\|_2
\end{equation}
where $x_i$ represents a sample in retain set ($x_i \in D_r$).
Notably, to maintain the existing refusal mechanism of LLMs, refusal responses $x_r$ to harmful queries are added into $D_\texttt{r}$.
Therefore, most $\{\mathcal{R}_{\mathcal{M}_\theta} \left(x_h \right)\}_{i=1}^N$ are finally optimized to refusal representations $\{\mathcal{R}_{\mathcal{M}_\theta} \left(x_r \right)\}\}_{i=1}^N$ under the joint effect of $\mathcal{L}_\texttt{e}$ and $\mathcal{L}_\texttt{r}$.

In summary, the overall loss function is a weighted combination of the three aforementioned loss functions:
\begin{equation}
    \mathcal{L}=c_r\mathcal{L}_r + c_e\mathcal{L}_e + c_s\mathcal{L}_s
\end{equation}
where $c_r$, $c_e$ and $c_s$ are adaptive loss coefficients following \cite{circuit_breaker, adaptive_loss}.
With the above optimization objective, X-Boundary can perform fine-grained optimization in the representation space to \textbf{reconcile the trade-off between defense robustness and the usability of LLMs}.
The overall optimization process of X-boundary is shown as Algorithm~\ref{algorithm} in Appendix~\ref{app:algorithm}.

\subsection{Theoretical Analysis of X-Boundary}
\label{sec:theoretical}
In this subsection, we theoretically analyze the convergence rate of LLM from the perspective of the optimal transport theory \cite{solomon2020k, chuang2021measuring, weed2017sharp}. Specifically, we theoretically prove that X-boundary enables a faster learning speed of feature learning, which is verified in Fig.~\ref{fig:loss_plot}.

\vspace{+1mm}
\noindent\textbf{Preliminaries: optimal transport and $k$-variance.} Wasserstein distance measures the distance between probability distributions on a metric space. Let $\mu$ and $\nu \in \Prob(\sR^d)$ denote two probability measures, the definition of $p$-Wasserstein distance with Euclidean cost function is
\begin{equation}
    \gW_p(\mu, \nu) = \inf_{\pi \in \Pi(\mu, \nu)} \left( \E_{(H,Q) \sim \pi} \|H-Q\|^p\right)^{1/p},
\end{equation}
where $\Pi(\mu, \nu) \subseteq \Prob(\sR^d \times \sR^d)$ represent the set of measure couplings and $\mu$ and $\nu$ denote their marginals, respectively. From the perspective of optimal transport, Wasserstein distances indicate the minimal cost of transforming the distribution $\mu$ to $\nu$. Typically, the Earth Mover distance is equivalent to the 1-Wasserstein distance. 

\begin{definition}[Wasserstein-$1$ $k$-variance]
Given a probability measure $\mu \in \Prob(\sR^d)$ and a parameter $k \in \sN$, the \emph{Wasserstein-$1$ $k$-variance} is
\begin{equation}
    \Var_{k}(\mu) =  \E_{S, \tilde{S} \sim \mu^k} \left[ \gW_1(\mu_S, \mu_{\tilde{S}} ) \right],
\end{equation}
where $\mu_S = \frac{1}{k}\sum_{i=1}^k \delta_{x_i}$ for $x_i \overset{\textnormal{i.i.d.}}{\sim} \mu $.
\end{definition}

$k$-variance measures structural properties of distribution beyond variance based on Wasserstein distances \cite{solomon2020k}. We theoretically analyze the learning trend of DNN feature representations, which can be measured by the convergence rate of $k$-variance following \cite{weed2017sharp, solomon2020k}.


\begin{proposition} (Proven in Appendix~\ref{ap:proof})
\label{prop_cluster}
 If $\phi_\# \mu$ is $(n, \Delta)$-clusterable, then for all $m \leq n(2\Delta)^{-2}$,
 \begin{equation}
     \Var_{m}(\phi_\# \mu) < 48\Delta.
 \end{equation}
 Given a distribution $\mu$, $(n, \Delta)$-clusterable means that $\textnormal{supp}(\mu)$ lies in the union of $n$ balls of radius at most $\Delta$.
\end{proposition}

 Proposition~\ref{prop_cluster} indicates that $\Var_{m}(\phi_\# \mu)$ is bounded by the radius $\Delta$, reflecting the concentration of the feature distribution. In this way, the proposed X-Boundary enables more clustered features (the smaller radius $\Delta$) and a faster learning speed (the smaller $k$-variance $\Var_{m}(\phi_\# \mu)$). 

\vspace{+1mm}
\noindent\textbf{Experimental Verification.} Fig.~\ref{fig:loss_plot} verifies that X-Boundary enables a faster learning speed of the training process. To this end, we fine-tune Llama-3-8B-Instruct and Qwen2.5-7B-Instruct following the settings in Section~\ref{sec:comparison}. Specifically, we set 0.1 and 0.55 of the training loss as thresholds to judge whether the training process has converged for Llama-3-8B-Instruct and Qwen2.5-7B-Instruct, respectively. Based on this, Fig.~\ref{fig:loss_plot} indicates that the proposed X-Boundary accelerates the converging process of 26.47\% and 18.29\% on Llama-3-8B-Instruct and Qwen2.5-7B-Instruct, respectively.
\begin{figure}[t]
\begin{center}
\centerline{\includegraphics[width=\columnwidth]{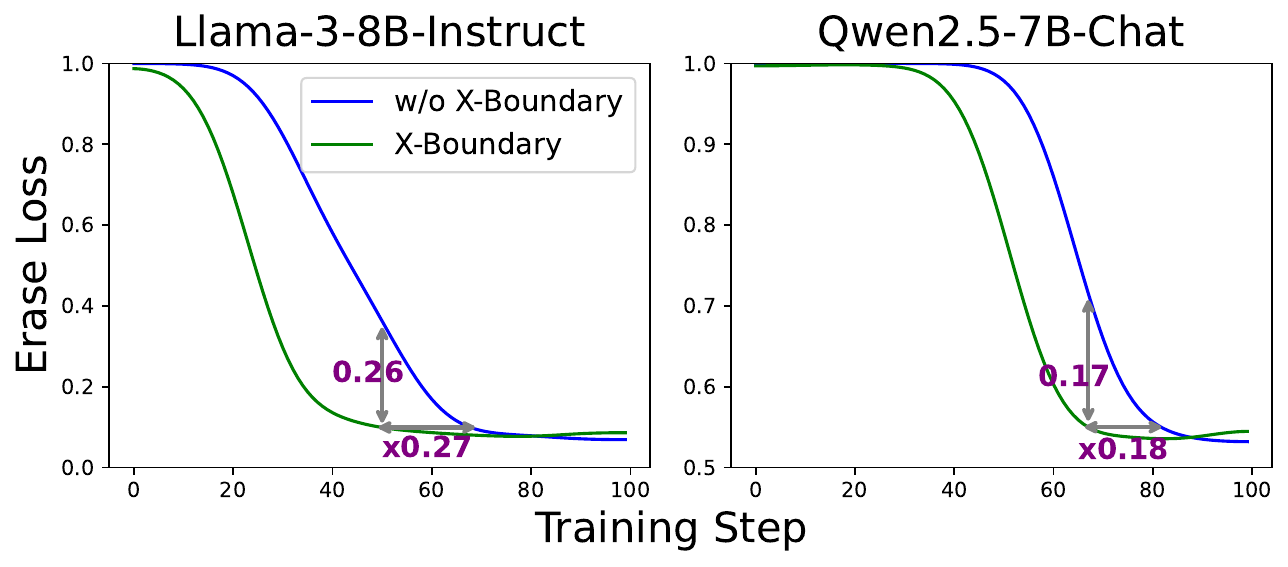}}
\vskip -0.1in
\caption{The training curves of X-Boundary and without X-Boundary on Llama-3-8B-Instruct and Qwen2.5-7B-Instruct.}
\label{fig:loss_plot}
\end{center}
\vskip -0.4in
\end{figure}

%% file: content/experiment.tex
\section{Experiments}
\begin{table*}[t]
\setlength{\tabcolsep}{1pt}
\centering
\resizebox{\linewidth}{!}{
\begin{tabular}{c|ccccccccccccc}
\midrule
\multirow{2.5}{*}{\textbf{Methods}} & \multicolumn{3}{c}{\textbf{Single-Turn ASR (\%) $\downarrow$}} & \multicolumn{3}{c}{\textbf{Multi-Turn ASR (\%) $\downarrow$}} & \multicolumn{4}{c}{\textbf{Over-Refusal Rate (\%) $\downarrow$}} & \multicolumn{3}{c}{\textbf{General Capability (\%) $\uparrow$}} \\
\cmidrule(lr){2-4} \cmidrule(lr){5-7} \cmidrule(lr){8-11} \cmidrule(lr){12-14}
~ & ~~GCG~~~ & PAIR & PAP & ActorAttack & RedQueen & Crescendo & XSTest & OKTest & OR-Bench & PHTest & MMLU & GSM8K & HumanEval \\
\midrule
\multicolumn{14}{c}{Llama-3-8B-Instruct} \\
\midrule
Vanilla & 31.00 & 18.00 & 15.00 & 58.50 & 25.00 & 34.00 & 6.80 & 9.00 & 8.00 & 13.67 & 68.30 & 79.08 & 59.18 \\
\midrule
SFT & 6.50 & 13.50 & \underline{1.50} & 19.50 & \textbf{0.50} & \textbf{8.00} & 27.20 & 42.33 & 22.00 & 57.33 & \underline{68.17} & 76.19 & 54.27 \\
DPO & 8.50 & \underline{11.00} & 3.00 & \underline{17.50} & 5.00 & 14.00 & 20.00 & 28.33 & 17.33 & 41.00 & 68.01 & 75.59 & 58.54 \\
GA & 18.00 & 11.50 & 3.50 & 38.50 & 1.50 & 12.00 & \underline{10.80} & \underline{15.00} & \underline{13.33} & \underline{35.33} & \textbf{68.25} & 77.86 & \textbf{62.20} \\
CB & \underline{2.00} & 12.00 & \textbf{1.00} & \textbf{16.50} & \textbf{0.50} & \underline{10.00} & 23.60 & 27.67 & 36.00 & 52.00 & 67.66 & \underline{78.47} & \underline{59.76} \\
\midrule
X-Boundary & \textbf{1.50} & \textbf{10.00} & \textbf{1.00} & \textbf{16.50} & \underline{1.00} & \underline{10.00} & \textbf{8.40} & \textbf{14.00} & \textbf{8.00} & \textbf{28.67} & 67.94 & \textbf{78.70} & \underline{59.76} \\
\midrule
\multicolumn{14}{c}{Qwen2.5-7B-Instruct} \\
\midrule
Vanilla & 76.00 & 48.50 & 51.50 & 76.00 & 39.50 & 62.00 & 6.00 & 19.33 & 1.67 & 5.67 & 74.26 & 80.67 & 81.71 \\
\midrule
SFT & 48.50 & 39.50 & 15.50 & 21.00 & 6.00 & 18.00 & 46.00 & 57.67 & 29.33 & 53.67 & \underline{74.30} & 76.42 & 77.44 \\
DPO & 46.50 & 48.00 & 21.50 & 38.00 & 12.00 & 24.00 & 21.60 & \underline{25.67} & \underline{11.67} & \underline{32.33} & 73.63 & \textbf{80.97} & \underline{80.49} \\
GA & 54.00 & 35.00 & \underline{9.50} & 38.00 & 21.00 & \textbf{12.00} & 58.33 & 70.00 & 67.67 & 85.33 & \textbf{74.58} & 80.43 & 79.27 \\
CB & \textbf{22.00} & \underline{27.50} & 10.50 & \textbf{15.50} & \textbf{5.50} & \textbf{12.00} & \underline{20.60} & 26.00 & 34.00 & 43.67 & 74.21 & 80.36 & \textbf{81.10} \\
\midrule
X-Boundary & \underline{23.00} & \textbf{26.00} & \textbf{8.50} & \underline{17.50} & \underline{7.50} & \underline{16.00} & \textbf{10.40} & \textbf{16.67} & \textbf{5.33} & \textbf{15.00} & 74.17 & \underline{80.52} & \textbf{81.10} \\
\midrule
\multicolumn{14}{c}{Mistral-7B-Instruct-v0.2} \\
\midrule
Vanilla & 83.50 & 60.50 & 61.00 & 70.00 & 49.50 & 40.00 & 10.00 & 21.00 & 4.33 & 13.00 & 59.98 & 45.34 & 34.76 \\
\midrule
SFT & 38.50 & 48.00 & 34.00 & 37.50 & 22.00 & 18.00 & 53.60 & 42.00 & 29.33 & 58.67 & 58.94 & 41.55 & 27.44 \\
DPO & 36.00 & 47.00 & 42.50 & 44.50 & 19.00 & 28.00 & \underline{25.20} & \underline{38.67} & \underline{20.33} & \underline{37.67} & 58.79 & 43.21 & \underline{34.76} \\
GA & 48.00 & \textbf{32.50} & \textbf{25.00} & 24.00 & \textbf{9.00} & \textbf{10.00} & 38.40 & 50.67 & 35.67 & 71.33 & \textbf{60.13} & 45.00 & \underline{34.76} \\
CB & \textbf{31.00} & 36.50 & 30.50 & \textbf{15.00} & \underline{11.50} & \underline{12.00} & 45.20 & 39.33 & 55.00 & 50.00 & \underline{59.91} & \textbf{46.63} & 33.54 \\
\midrule
X-Boundary & \underline{34.50} & \underline{35.00} & \underline{30.00} & \underline{16.00} & 13.50 & 14.00 & \textbf{19.20} & \textbf{23.33} & \textbf{10.34} & \textbf{26.33} & 59.83 & \underline{45.34} & \textbf{36.59} \\
\midrule
\end{tabular}}
\setlength{\abovecaptionskip}{0.1in} 
\caption{Comparison of existing defense methods and X-Boundary.}
\label{tab:main_results}
\vspace{-8pt}
\end{table*}
\subsection{Experimental Settings}
To ensure fairness in comparison and consistency in experimental settings, we implement four baseline methods and X-Boundary on Llama-3-8B-Instruct, Qwen2.5-7B-Instruct, and Mistral-7B-Instruct-v0.2, and evaluate them using HarmBench dataset~\cite{harmbench} and the metrics described in Section~\ref{sec:comparison}.
Additionally, to assess the effectiveness of X-Boundary across different sizes of LLMs, we implement it on Qwen2.5-14B-Instruct.
To construct the Separate Set, we sample 500 boundary-safe prompts from OR-Bench-80K \cite{orbench}, which have been filtered to avoid data contamination with the test set of OR-Bench. 
Next, we use GPT-4o to generate safe and helpful responses for these prompts, thus we get boundary-safe QA pairs.
The retain set consists of boundary-safe QA pairs, UltraChat \cite{ultrachat}, and refusal data points generated by the trained LLMs themselves.
The erase set includes the harmful QA pairs for single-turn defense used in \citet{circuit_breaker} and the harmful QA pairs for multi-turn defense described in Section~\ref{sec:comparison}.
Evaluation and implementation details of X-Boundary are listed in Appendix~\ref{app:eval} and~\ref{app:x_training}, respectively.
\subsection{Main Results}

\vspace{+1mm}
\noindent\textbf{The explicit formulation for boundary contributes to the precise distinction between harmful and safe representations.}
To investigate the effect of the explicit formulation for distinction boundary, we visualize the representation distribution of X-Boundary and without X-Boundary. 
Fig.~\ref{figs:t-sne} shows that, without X-Boundary, the boundary-safe representations close to harmful representations are mistakenly regarded as harmful ones.
%
This demonstrates that LLMs fail to learn a boundary that exactly distinguishes safe and harmful representations, which supports our motivation of explicitly formulating the distinction boundary.
With X-Boundary, harmful representations and boundary-safe representations are clearly separated as shown in Fig.~\ref{figs:t-sne}, verifying that the proposed explicit formulation contributes to establishing a precise distinction boundary. 
Please refer to Appendix~\ref{app:repe_angle} and~\ref{app:complete_tsne} for more detailed visualization of the representation distribution.

\vspace{+1mm}
\noindent\textbf{X-Boundary maintains the lowest ORR while achieving SOTA defense against both single-turn and multi-turn jailbreaks.}
With a precise distinction boundary, X-Boundary relatively reduces single-turn and multi-turn ASR by more than 40\% while maintaining the increase in ORR on OKTest within 5\% across three LLMs, as shown in Table~\ref{tab:main_results}.
Specifically, on Llama-3-8B-Instruct, CB and X-Boundary both achieve the lowest ASR against ActorAttack, but X-Boundary demonstrates an average ORR that is lower by 20.05\%. 
Similarly, on Qwen2.5-7B-Instruct, X-Boundary’s average ORR is 58.50\% lower than GA, which achieves the lowest ASR against Crescendo.

\vspace{+1mm}
\noindent\textbf{X-Boundary rarely declines general capability.}
Table~\ref{tab:main_results} shows that the decline of general capabilities caused by X-Boundary is generally no more than 0.5\% compared to vanilla models, across the domains of general knowledge, mathematical ability, and coding ability.
In contrast to SFT, which causes a 7\% reduction in coding ability for Mistral-7B-Instruct-v0.2, X-Boundary achieves a lower ASR without compromising coding capability.
More evaluations of single-turn defense are listed in Appendix~\ref{app:single_turn_asr}.

\vspace{+1mm}
\noindent\textbf{X-Boundary successfully alleviates the trade-off between robustness and usability.}
As a supplement to Table~\ref{tab:main_results}, Fig.~\ref{fig:trade_off} intuitively illustrates the trade-off between ASR against jailbreaks and ORR.
Considering the two metrics comprehensively, X-Boundary appears in the lower-left corner of Fig.~\ref{fig:trade_off} and increases the hypervolume, \ie the volume of the dominated space between the Pareto front and a predefined reference point, by 13.13\% and 10.03\% in OKTest and PHTest, respectively. 
The results indicate that X-Boundary significantly advances the Pareto frontier and mitigates the trade-off between ASR and ORR compared to the baseline methods.
In the same way, Fig.~\ref{fig:asr_utility} in Appendix~\ref{app:asr_utility} demonstrates that X-Boundary also achieves a win-win outcome with robust defense and strong general capability.
For specific cases of the defense performance and usability preservation of X-Boundary, please refer to Appendix~\ref{app:case_study}.

\vspace{+1mm}
\noindent\textbf{X-Boundary is effective across different sizes of LLMs.}
Table~\ref{tab:large_model_results} in Appendix~\ref{app:large_model_results} shows that, on Qwen2.5-14B-Instruct, X-Boundary relatively reduces the ASR by more than 60\%, while keeping the increase in ORR within 5\% compared to the vanilla model.
Although X-Boundary and CB achieve comparable ASRs, the ORR of X-Boundary is approximately 40\% lower than that of CB.
Compared with the performance on Qwen2.5-7B-Instruct, those of X-Boundary on Qwen2.5-14B-Instruct is stable and has not decreased.
\begin{figure}[t]
\begin{center}
\centerline{\includegraphics[width=\columnwidth]{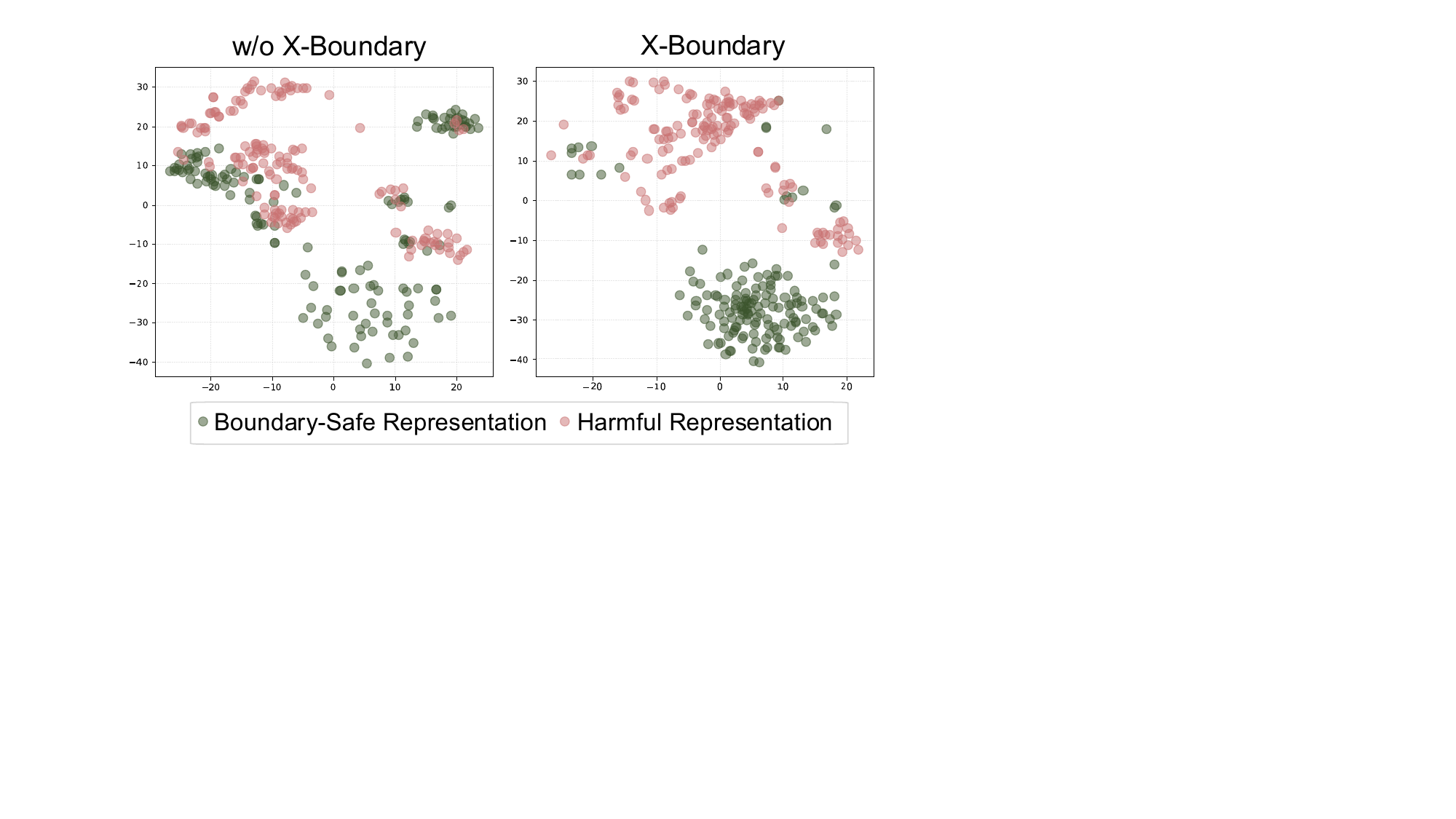}}
\setlength{\abovecaptionskip}{0.1in} 
\caption{Visualization of the representation distribution of X-Boundary and without X-Boundary.} 
\label{figs:t-sne}
\end{center}
\vskip -0.3in
\end{figure}
\renewcommand{\thesubfigure}{}
\begin{figure}[t]
\begin{center}
\centerline{\includegraphics[width=\columnwidth]{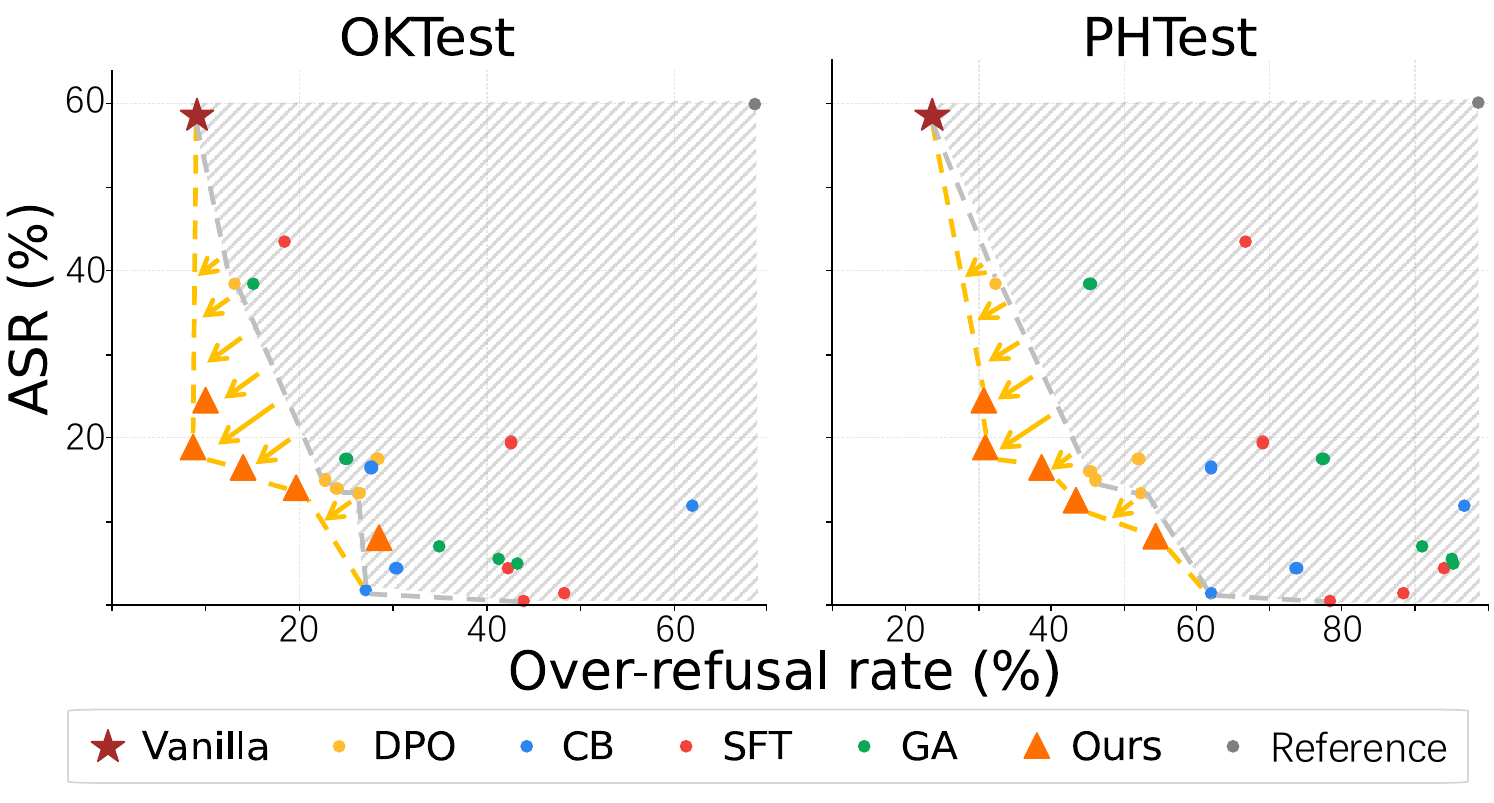}}
\setlength{\abovecaptionskip}{0.1in} 
\caption{The trade-off between ASR of jailbreaks and ORR. The data points are collected by sampling and evaluating every 100 training steps.}
\label{fig:trade_off}
\end{center}
\vskip -0.45in
\end{figure}
\begin{figure}[t]
\begin{center}
\centerline{\includegraphics[width=\columnwidth]{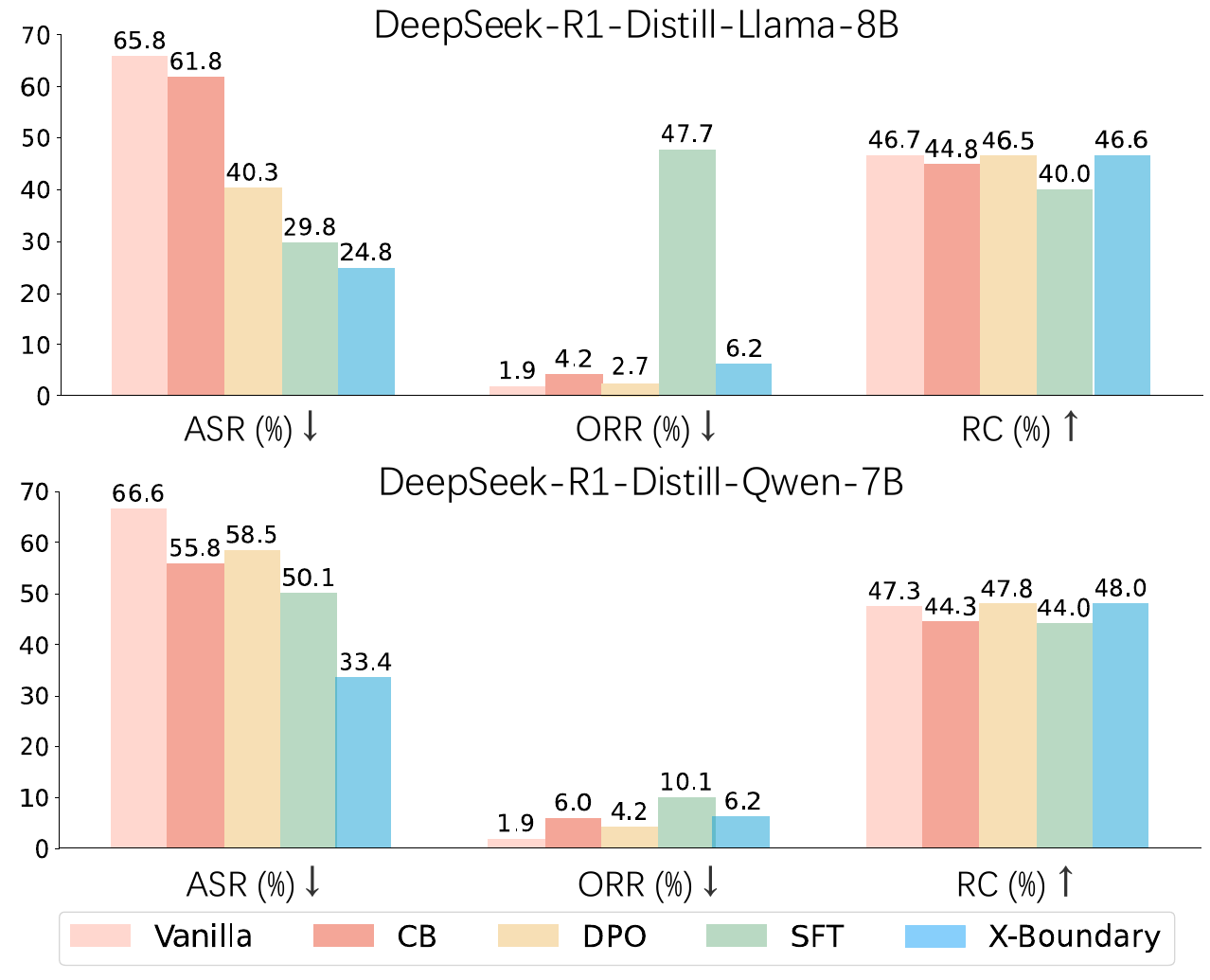}}
\setlength{\abovecaptionskip}{0.1in} 
\caption{Comparison of existing defense methods and X-Boundary on DeepSeek-R1-distilled models.} 
\label{fig:reasoning_performance}
\end{center}
\vskip -0.3in
\end{figure}
\begin{table*}[t]
\setlength{\tabcolsep}{2pt}
\centering
\resizebox{2\columnwidth}{!}{
\begin{tabular}{c|cccc|ccccccccccccc}
\midrule
\multirow{2.5}{*}{\textbf{Models}} & \multirow{2.5}{*}{A} & \multirow{2.5}{*}{B} & \multirow{2.5}{*}{C} & \multirow{2.5}{*}{D} & \multicolumn{3}{c}{\textbf{Multi-Turn ASR (\%) $\downarrow$}} & \multicolumn{4}{c}{\textbf{Over-Refusal Rate (\%) $\downarrow$}} & \multicolumn{3}{c}{\textbf{General Capability (\%) $\uparrow$}} \\
\cmidrule(lr){6-8} \cmidrule(lr){9-12} \cmidrule(lr){13-15}
~ & & & & & ActorAttack & RedQueen & Crescendo & XSTest & OKTest & OR-Bench & PHTest & MMLU & GSM8K & HumanEval \\
\midrule
Vanilla & & & & & 76.00 & 39.50 & 62.00 & 6.00 & 19.33 & 1.67 & 5.60 & 74.26 & 80.67 & 81.71 \\
\cmidrule(lr){1-15}
(a) & $\checkmark$ & & & & 63.00 & 11.50 & 30.00 &9.20 & 19.00 & 6.66 & 14.66 & 74.19 & 80.14 & 82.32\\
(b) & $\checkmark$ & $\checkmark$ & & & 15.50 & 5.50 & 12.00 & 20.40 & 26.00 & 34.00 & 43.67 & 74.21 & 80.36 & 81.10\\
(c) & $\checkmark$ & $\checkmark$ & $\checkmark$ & & 15.50 & 7.00 & 16.00 & 18.00 & 28.33 & 6.33 & 25.00 & 74.20 & 80.36 & 81.71\\
\cmidrule(lr){1-15}
X-Boundary & $\checkmark$ & $\checkmark$ & $\checkmark$ & $\checkmark$ & 17.50 & 7.50 & 16.00 & 10.40 & 16.67 & 5.33 & 15.00 & 74.17 & 80.52 & 81.10\\
\midrule
\end{tabular}}
\setlength{\abovecaptionskip}{0.1in} 
\caption{Ablation study on Qwen2.5-7B-Instruct. In this table, A represents single-turn defense data, B represents multi-turn defense data, C represents boundary-safe data, and D represents the separate loss $\mathcal{L}_{\texttt{s}}$.}
\label{tab:ablation}
\vspace{-8pt}
\end{table*}
\subsection{Performance on Large Reasoning Models}
Recently, several studies~\cite{safechain,R1_assessment} have highlighted significant safety risks in the outputs of large reasoning models (LRMs), particularly during the thinking process.
Enhancing the security of LRMs, such as DeepSeek-R1~\cite{guo2025deepseek}, has become an urgent priority.
In this section, we evaluate the performance of X-Boundary and baseline methods on two LRMs: DeepSeek-R1-distilled-LLaMA-8B and DeepSeek-R1-distilled-Qwen-7B.
The evaluation of defense performance and over-refusal adopts the same datasets and metrics as Section~\ref{sec:comparison} described.
To assess general capability, we replace the previous datasets with more challenging benchmarks that test reasoning capability (RC), namely AIME2024, GPQA, and LiveCodeBench.
The detailed evaluation settings and analysis of RC are listed in Appendix~\ref{app:reasoning_ability}.

As shown in Fig.~\ref{fig:reasoning_performance}, both CB and DPO exhibit marginal defense effectiveness on LRMs, reducing the average ASR by only around 10\% on the Distilled-Qwen model.
Although SFT still demonstrates robust defense on LRMs, it causes a degradation of over 5\% in RC and leads to a significant increase in the average ORR.
In contrast, X-Boundary achieves outstanding defense performance while maintaining the average ORR below 10\% and preserving 99\% RC.
%
This result may be attributed to the theoretical analysis in Section~\ref{sec:theoretical}, which suggests that X-Boundary reduces the difficulty of training and facilitates faster convergence within the complex representation space of LRMs.


\subsection{Ablation Study}
\label{sec:ablation}
We conduct ablation studies on the impact of multi-turn defense data, boundary-safe data, and separate loss. The results are illustrated in Table~\ref{tab:ablation}. Ablation Studies on Llama-3-8B-Instruct and Mistral-7B-Instruct-v0.2 are shown in Appendix~\ref{app:ablation}. Please see Appendix~\ref{app:loss_ablation} and~\ref{app:sensitive_analysis} for ablation studies on three terms of loss and sensitivity analysis on hyper-parameters $\alpha$ and $\beta$, respectively.

\vspace{+1mm}
\noindent\textbf{Multi-turn defense data contribute to the reduction of ASR but intensify the over-refusal problem.} 
With the multi-turn defense data described in Section \ref{sec:comparison} added into the erase set, the ASR of ActorAttack is reduced from 63.00\% to 15.50\% on Qwen2.5-7B-Instruct. 
However, the ORRs in OR-Bench and PHTest increase by about 30.00\%.

\vspace{+1mm}
\noindent\textbf{Boundary-safe data can partially mitigate the over-refusal issue.}
Boundary-safe QA pairs added to the retain set significantly reduce the ORR on OR-Bench and PHTest but show limited effectiveness on XSTest and OKTest. 
This may be because the boundary-safe QA pairs are synthesized by LLMs, leading to effectiveness on OR-Bench and PHTest, which also use synthetic data for testing. 
In contrast, the test queries in XSTest and OKTest are manually crafted and may differ in distribution from the synthetic data, making it difficult to achieve effective generalization.

\vspace{+1mm}
\noindent\textbf{Simply adjusting the size of boundary-safe data can not effectively balance ASR and ORR.}
%
Increasing the size of boundary-safe data can reduce the ORR, but it also leads to a sharp increase in ASR against jailbreaks.
Please see Appendix \ref{app:set_size_effect} for more detailed results.

\vspace{+1mm}
\noindent\textbf{Separate loss can further reduce the ORR.}
Unlike simply adding boundary-safe data, separate loss markedly reduces the ORR on both manually crafted and synthetically constructed benchmarks.
Since the boundary-safe data shares the same source as OR-Bench, simply adding data is sufficient to reduce the ORR to a very low level, leaving little room for separate loss to make a noticeable impact. 
However, in the other three benchmarks, separate loss further reduces the ORR by an average of 9.75\%.

%% file: content/conclusion.tex
\section{Conclusion}
In this paper, we comprehensively compare existing jailbreak defense methods and reveal the trade-off between the robustness of defense and LLM usability. 
We analyze this issue from the perspective of LLMs' feature space, and conclude that previous methods fail to learn a precise boundary that distinguishes safe and harmful representations without an explicit formulation.
To address this issue, we propose X-Boundary to push harmful representations away from safe representations through explicit loss functions and obtain a clear distinction boundary. 
Such distinction boundary enables the consequential removal of harmful representations without disrupting safe ones, thereby achieving a balance between robustness against jailbreaks and LLM usability.
%
%
We think that X-Boundary can offer a more efficient and fine-grained defense for LLMs, improving the deployment of robust AI systems in real-world applications.

%% file: content/limitation.tex
\section*{Limitations}
This paper has several limitations. 
First, although we analyze the underlying causes of the trade-off between defense robustness and LLM usability and propose a post-training method to achieve a mutually beneficial outcome, we have not yet thoroughly investigated how to fundamentally resolve this issue during the pre-training stage, as the pre-training processes of these LLMs are closed-source.
Second, due to its reliance on representation-level intervention, X-Boundary is not applicable to black-box models, thereby restricting its use in some practical settings.

\section*{Ethical considerations}
This work aims to advance the field of large language models (LLMs) safety alignment by proposing X-Boundary, a method that maintains state-of-the-art performance in multi-turn jailbreak attack defenses while effectively mitigating the over-safety problem.
All the training data and reproduced defense methods we used are open-source and consistent with their intended use, with proper citations to their original sources.
We do not consider that this method will directly lead to severe negative consequences for societal development. However, we must be aware that malicious actors could exploit various approaches to induce LLMs to generate misleading or harmful content. Besides, training data containing some harmful or offensive questions and answers pose a risk of malicious use and potential harm. Therefore, we expect that future research will focus on enhancing content moderation mechanisms and setting up ethical usage protocols to effectively reduce potential risks. 

\section*{Acknowledgements}
This work is supported by Shanghai Artificial Intelligence Laboratory. And we would like to express our gratitude to our collaborators for their efforts.

%% file: content/appendix.tex
\appendix
\input{content/related_work}
\section{The Optimization Process of X-Boundary}
\label{app:algorithm}
The optimization process of X-Boundary is shown as Algorithm~\ref{algorithm}.
\begin{algorithm}[tb]
\caption{The optimization process of X-Boundary}
\label{algorithm}
\begin{algorithmic}[1]
    \REQUIRE Original frozen model $\mathcal{M}_{\texttt{ref}}$, model $\mathcal{M}_{\theta}$ with parameters $\theta$ to be optimized, a function $\mathcal{R}$ that extracts representation from a model on a batch of inputs, a erase dataset $\mathcal{D}_e$, a retain dataset $\mathcal{D}_r$, a boundary dataset $D_b$, number of optimization steps $T$, hyperparameters $\alpha$ and~$\beta$, batch size $n$ 
   \begin{spacing}{1.2}
   \FOR{$t=1$ {\bfseries to} $T$}
   \STATE Sample $\{x_i\}_{i=1}^n \sim \mathcal{D}_r$, $\{x_i^h\}_{i=1}^n \sim \mathcal{D}_e$
   \STATE Sample $\{(x_i^b,x_i^r)\}_{i=1}^n \sim \mathcal{D}_b$ 
   \STATE $c_r=\alpha \frac{t}{\beta}$, $c_e=c_s=\alpha(1-\frac{t}{\beta})$
   \STATE $\mathcal{L}_\text{r} = \frac{1}{n}\sum_{i=1}^{n} \left\| \mathcal{R}_{\mathcal{M}_\theta} \left(x_i \right) - \mathcal{R}_{\mathcal{M}_{\texttt{ref}}} \left(x_i \right) \right\|_2$ 
   \STATE {\scriptsize $\mathcal{L}_\texttt{e} =\frac{1}{n}\sum_{i=1}^{n}\texttt{ReLU} \left(\texttt{cos}\left(\mathcal{R}_{\mathcal{M}_{\theta}} \left(x_i^h \right), \mathcal{R}_{\mathcal{M}_{\texttt{ref}}} \left(x_i^h \right)\right) \right)$} 
   \STATE {\scriptsize $\mathcal{L}_\texttt{s} = \frac{1}{n}\sum_{i=1}^{n}\texttt{ReLU} \left(\texttt{cos}\left(\mathcal{R}_{\mathcal{M}_\theta} \left(x_i^r \right), \mathcal{R}_{\mathcal{M}_{\texttt{ref}}} \left(x_i^b \right) \right) \right)$} 
   \STATE $\mathcal{L}=c_r\mathcal{L}_r + c_e\mathcal{L}_e + c_s\mathcal{L}_s$ 
   \STATE Update parameters $\theta$ to minimize $\mathcal{L}$
   \ENDFOR
   \end{spacing}
\end{algorithmic}
\end{algorithm}

\section{Additional Results}
\subsection{Evaluation of Existing Over-Refusal Mitigation Methods}
\label{app:over_refusal_mitigation}
To further investigate the trade-off issue, we implement three over-refusal mitigation methods: system prompt (SP)~\cite{oktest}, Self-Contrastive Decoding (Self-CD)~\cite{oktest}, and vector ablation (VA)~\cite{wang2024surgical}. 
Table~\ref{tab:over_refusal_mitigation_vanilla} shows that these methods are effective on the vanilla model (Qwen2.5-7B-Instruct) and do not lead to a significant increase in ASR. 
However, as shown in Table~\ref{tab:over_refusal_mitigation_defense}, their impact on reducing ORR is less noticeable in models fine-tuned with defense methods, and they substantially weaken the defense effectiveness. 
Furthermore, both Self-CD and VA depend on refusal vectors or refusal tokens, which are ineffective for methods like CB that do not use a fixed refusal template.

\begin{table*}[!ht]
\setlength{\tabcolsep}{4pt}
\centering
\resizebox{\textwidth}{!}{
\begin{tabular}{c|ccccccccc}
\midrule
\multirow{2.5}{*}{\textbf{Methods}} & \multicolumn{2}{c}{\textbf{Attack Success Rate (\%) $\downarrow$}} & \multicolumn{4}{c}{\textbf{Over-Refusal Rate (\%) $\downarrow$}} & \multicolumn{3}{c}{\textbf{General Capability (\%) $\uparrow$}} \\
\cmidrule(lr){2-3} \cmidrule(lr){4-7} \cmidrule(lr){8-10}
& DirectRequest & ActorAttack & XSTest & OKTest & OR-Bench & PHTest & MMLU & GSM8K & HumanEval \\
\midrule
Qwen2.5-7B-Instruct & 26.25 & 76.00 & 6.00 & 19.33 & 1.67 & 5.67 & 74.26 & 80.67 & 81.71 \\
\midrule
+SP & 26.67 & 78.50 & 2.80 & 9.33 & 1.67 & 3.67 & 74.30 & 80.97 & 81.10 \\
+Sefl-CD & 28.33 & 78.00 & 2.80 & 9.33 & 1.00 & 4.33 & 74.21 & 80.52 & 82.93 \\
+VA & 27.92 & 75.50 & 4.20 & 11.00 & 1.33 & 3.00 & 74.58 & 80.36 & 81.71 \\
\midrule
\end{tabular}}
\setlength{\abovecaptionskip}{0.1in} 
\caption{Performance of existing over-refusal mitigation methods on Qwen2.5-7B-Instruct.}
\label{tab:over_refusal_mitigation_vanilla}
\vspace{-6pt}
\end{table*}

\begin{table*}[!ht]
\setlength{\tabcolsep}{4pt}
\centering
\resizebox{\textwidth}{!}{
\begin{tabular}{c|ccccccccc}
\midrule
\multirow{2.5}{*}{\textbf{Methods}} & \multicolumn{2}{c}{\textbf{Attack Success Rate (\%) $\downarrow$}} & \multicolumn{4}{c}{\textbf{Over-Refusal Rate (\%) $\downarrow$}} & \multicolumn{3}{c}{\textbf{General Capability (\%) $\uparrow$}} \\
\cmidrule(lr){2-3} \cmidrule(lr){4-7} \cmidrule(lr){8-10}
& DirectRequest & ActorAttack & XSTest & OKTest & OR-Bench & PHTest & MMLU & GSM8K & HumanEval \\
\midrule
Qwen2.5-7B-Instruct & 26.25 & 76.00 & 6.00 & 19.33 & 1.67 & 5.67 & 74.26 & 80.67 & 81.71 \\
\midrule
+SFT & 5.42 & 21.00 & 46.00 & 57.67 & 29.33 & 53.67 & 74.30 & 76.42 & 77.44 \\
+SFT+SP & 6.25 & 41.00 & 37.20 & 47.00 & 26.00 & 44.00 & 74.17 & 75.51 & 78.66 \\
+SFT+Sefl-CD & 6.00 & 28.50 & 44.80 & 52.67 & 28.33 & 54.00 & 73.63 & 77.94 & 79.27 \\
+SFT+VA & 8.75 & 43.50 & 23.60 & 41.33 & 23.67 & 40.00 & \textbf{74.58} & 77.94 & 78.66 \\
\midrule
+CB & 1.67 & \textbf{15.50} & 20.60 & 26.00 & 34.00 & 43.67 & 74.21 & 80.36 & 81.10 \\
+CB+SP & 2.92 & 27.00 & 20.20 & 27.33 & 35.67 & 42.00 & 74.21 & 80.43 & 80.49 \\
+CB+Sefl-CD & 4.58 & 26.50 & 24.80 & 25.00 & 37.33 & 46.33 & 74.30 & 80.52 & 79.88 \\
+CB+VA & 2.08 & 20.50 & 19.20 & 24.00 & 33.67 & 41.33 & 73.67 & 80.43 & 80.49 \\
\midrule
X-Boundary & \textbf{1.25} & 17.50 & \textbf{10.40} & \textbf{16.67} & \textbf{5.33} & \textbf{15.00} & 74.17 & \textbf{80.52} & \textbf{81.10} \\
\midrule
\end{tabular}}
\setlength{\abovecaptionskip}{0.1in} 
\caption{Performance of existing over-refusal mitigation methods on Qwen2.5-7B-Instruct fine-tuned with defense methods.}
\label{tab:over_refusal_mitigation_defense}
\vspace{-6pt}
\end{table*}

\subsection{Performance on Qwen2.5-14B-Instruct}
\label{app:large_model_results}
Table~\ref{tab:large_model_results} shows that X-Boundary also achieves SOTA defense and the lowest ORR on Qwen2.5-14B-Instruct.

\begin{table*}[!ht]
\setlength{\tabcolsep}{4pt}
\centering
\resizebox{\textwidth}{!}{
\begin{tabular}{c|cccccccccccc}
\midrule
\multirow{2.5}{*}{\textbf{Methods}} & \multicolumn{4}{c}{\textbf{Attack Success Rate (\%) $\downarrow$}} & \multicolumn{4}{c}{\textbf{Over-Refusal Rate (\%) $\downarrow$}} & \multicolumn{3}{c}{\textbf{General Capability (\%) $\uparrow$}} \\
\cmidrule(lr){2-5} \cmidrule(lr){6-9} \cmidrule(lr){10-12}
& DirectRequest & ActorAttack & RedQueen & Crescendo & XSTest & OKTest & OR-Bench & PHTest & MMLU & GSM8K & HumanEval \\
\midrule
Vanilla & 15.83 & 71.50 & 63.50 & 36.00 & 4.00 & 10.00 & 1.33 & 4.00 & 80.06 & 82.49 & 79.88 \\
\midrule
SFT & 7.08 & 52.00 & 10.00 & 16.00 & 43.60 & 51.33 & 31.33 & 62.67 & 79.58 & 82.18 & 81.71 \\
DPO & 8.33 & 54.50 & 45.00 & 32.00 & 6.40 & 14.00 & \textbf{2.67} & 8.67 & 78.58 & \textbf{83.32} & 81.10  \\
CB & 3.33 & \textbf{23.50} & \textbf{4.50} & \textbf{8.00} & 43.60 & 51.33 & 32.00 & 64.33 & \textbf{79.64} & 82.56 & \textbf{82.93} \\
\midrule
X-Boundary & \textbf{2.91} & 25.00 & 5.00 & 12.00 & \textbf{5.20} & \textbf{13.67} & 4.00 & \textbf{8.33} & 79.52 & 82.18 & 81.10 \\
\midrule
\end{tabular}}
\setlength{\abovecaptionskip}{0.1in} 
\caption{Comparison of existing defense methods and X-Boundary on Qwen2.5-14B-Instruct.}
\label{tab:large_model_results}
\vspace{-6pt}
\end{table*}

\subsection{Defense Performance Against Single-Turn Jailbreak Attacks}
\label{app:single_turn_asr}
We evaluate the robustness of X-Boundary and baseline methods against seven single-turn jailbreak attacks, \ie GCG~\cite{zou2023GCG}, PAIR~\cite{chao2023PAIR}, PAP~\cite{zeng2024PAP}, AutoDAN~\cite{autodan}, Obfuscation~\cite{zhang2024obfuscation}, Spliting~\cite{kang2023spliting}, and Multilingual~\cite{yong2023translation}.
Table \ref{tab:single_turn_asr} shows X-Boundary can effectively reduce the ASR of these attacks.

\begin{table*}[h]
\centering
\resizebox{0.9\textwidth}{!}{
\begin{tabular}{c|cccccccc}
\midrule
Methods & DirectRequest & GCG & PAIR & PAP & AutoDAN & Obfuscation & Splitting & Multilingual \\
\midrule
Vanilla & 11.67 & 31.00 & 18.00 & 15.00 & 4.50 & 12.00 & 15.00 & 3.00 \\
\midrule
SFT & 1.25 & 6.50 & 13.50 & 1.50 & 0.50 & 2.00 & 7.00 & \textbf{0.00} \\
DPO & \textbf{0.83} & 8.50 & 11.00 & 3.00 & \textbf{0.00} & 4.00 & \textbf{1.00} & \textbf{0.00} \\
GA & 5.00 & 18.00 & 11.50 & 3.50 & 1.50 & 9.50 & 7.00 & 1.00 \\
CB & 1.67 & 2.00 & 12.00 & \textbf{1.00} & \textbf{0.00} & \textbf{0.00} & 2.00 & \textbf{0.00} \\
\midrule
X-Boundary & 1.25 & \textbf{1.50} & \textbf{10.00} & \textbf{1.00} & \textbf{0.00} & 0.50 & 3.00 & \textbf{0.00} \\
\midrule
\end{tabular}}
\setlength{\abovecaptionskip}{0.1in} 
\caption{The ASR of seven single-turn jailbreak attacks after using existing defense methods and X-Boundary.}
\label{tab:single_turn_asr}
\vspace{-10pt}
\end{table*}

\subsection{The Effect of Defense Methods on the LLMs' Reasoning Ability}
\label{app:reasoning_ability}
Large reasoning models often rely on generating lengthy reasoning paths for inference. Therefore, we conducted a statistical analysis of the output length of large reasoning models employing various defense mechanisms. As shown in Table~\ref{tab:reasoning_model_len}, while X-Boundary does not lead to a degradation in general capability, it results in shorter output lengths, which may indirectly impact reasoning performance. Exploring strategies to prevent the reduction in output length represents a promising direction for future research.

\begin{table*}[t]
\setlength{\tabcolsep}{4pt}
\centering
\resizebox{0.8\textwidth}{!}{ 
\begin{tabular}{cc|cccccc} 
\midrule
\multirow{2}{*}{\textbf{Models}} & \multirow{2}{*}{\textbf{Methods}} 
& \multicolumn{2}{c|}{\textbf{AIME2024}} 
& \multicolumn{2}{c|}{\textbf{GPQA}} 
& \multicolumn{2}{c}{\textbf{LiveCode}} \\
\cmidrule(lr){3-4} \cmidrule(lr){5-6} \cmidrule(lr){7-8}
& & pass@1 & Length (Avg.) & pass@1 & Length (Avg.) & pass@1 & Length (Avg.) \\
\midrule
\multirow{5}{*}{\makecell{\textbf{DeepSeek-}\\\textbf{R1-Distill-}\\\textbf{Llama-8B}}} 
& Vanilla & 50.00 & 15672.07 & 50.00 & 8910.93 & 40.00 & 6457.43 \\
\cmidrule(lr){2-8}
& SFT & 44.95 & 13678.53 & 40.00 & 8699.93 & 35.10 & 6804.28 \\
& DPO & 46.97 & 15716.27 & 50.00 & 8489.33 & 42.40 & 6301.96 \\
& CB & 46.97 & 15488.23 & 46.97 & 9088.78 & 40.65 & 6479.9 \\
\cmidrule(lr){2-8}
& X-Boundary & 50.00 & 13310.90 & 50.00 & 8233.20 & 39.86 & 6498.04 \\
\midrule
\multirow{5}{*}{\makecell{\textbf{DeepSeek-}\\\textbf{R1-Distill-}\\\textbf{Qwen-7B}}} 
& Vanilla & 53.33 &11046.63 & 48.99 & 8592.54 & 39.76 & 6683.22 \\
\cmidrule(lr){2-8}
& SFT & 46.67 & 13844.87 & 48.99 & 8176.29 & 36.44 & 6825.17 \\
& DPO & 53.33 & 12063.57 & 50.00 & 8344.05 & 40.08 & 6694.74 \\
& CB & 46.97& 12609.93 & 46.97 & 8356.40 & 39.33 & 6536.76\\
\cmidrule(lr){2-8}
& X-Boundary & 53.33 & 12959.73 & 50.51 & 8237.67 & 40.02 & 6583.29 \\
\midrule
\end{tabular}}
\setlength{\abovecaptionskip}{0.1in} 
\caption{Comparison of pass@1 accuracy and average output token length across different defense methods on reasoning model}
\label{tab:reasoning_model_len}
\vspace{-2pt}
\end{table*}

\subsection{The Trade-Off between Robustness and General Capability}
\label{app:asr_utility}
Fig. \ref{fig:asr_utility} intuitively shows the trade-off between the ASR against multi-turn jailbreaks and the decline of general capability.
As the training process advances, the ASR steadily decreases, while the decline in code and math capability progressively increases.
%
X-Boundary lies in the lower-left corner of the plots, demonstrating that it achieves a win-win outcome with robust defense and strong general capability.

\renewcommand{\thesubfigure}{}
\begin{figure*}[!ht]
	\begin{center}
		\subfigure[]{
			\centering
            \includegraphics[width=0.8\columnwidth]{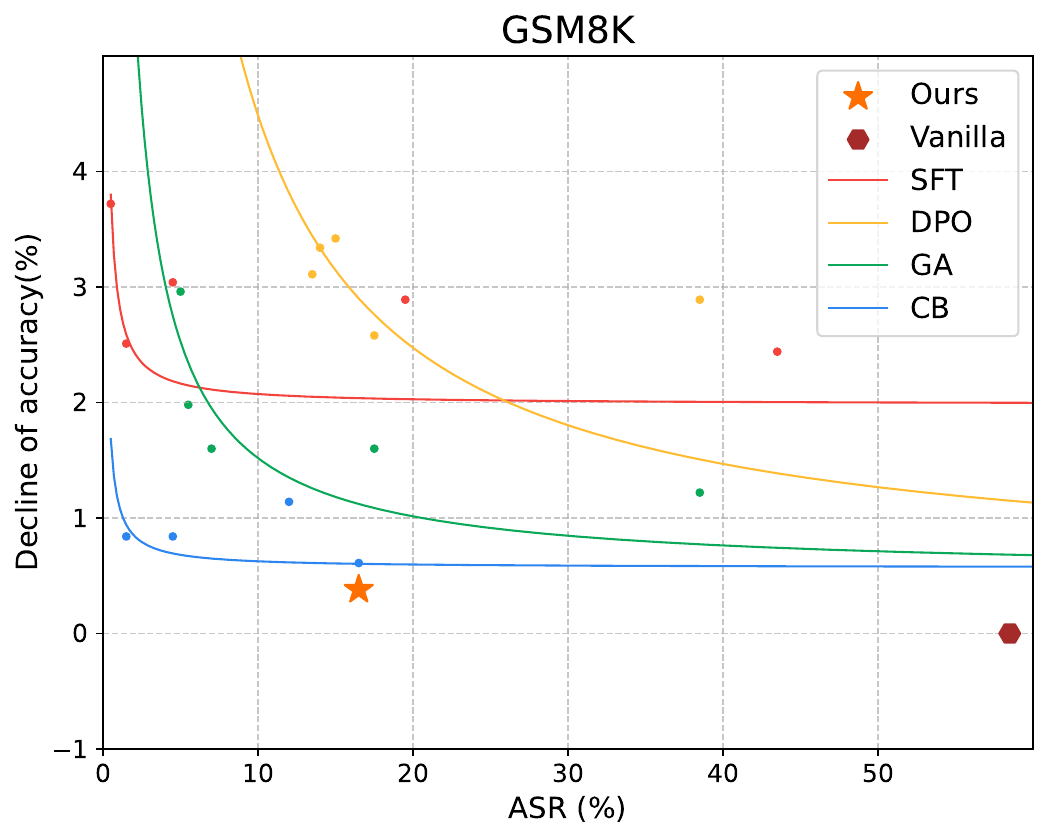}      
		}
		\subfigure[]{
			\centering
			\includegraphics[width=0.8\columnwidth]{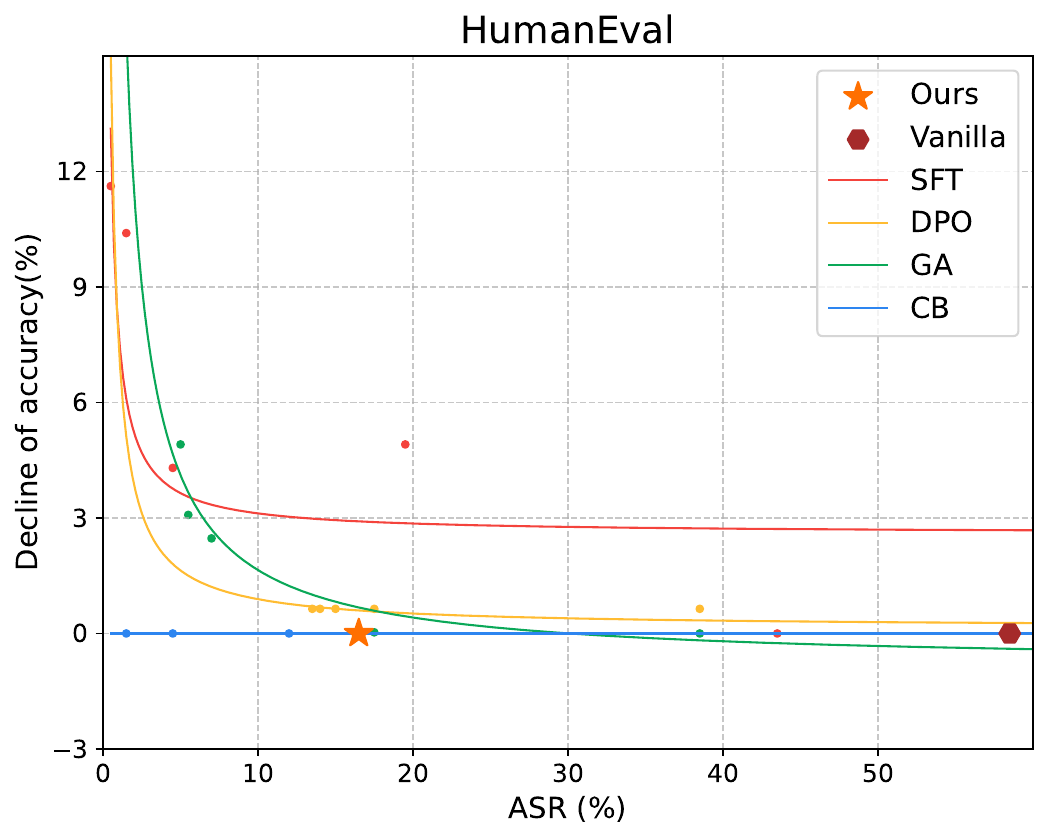}     
		}\vspace{-10mm}
	\end{center}
	\caption{The trade-off between ASR of multi-turn jailbreak and general capability on Llama-3-8B-Instruct. The data points were collected by sampling and evaluating at every 100 training steps.}
    \label{fig:asr_utility}
\end{figure*}

\subsection{Ablation Studies on Three Loss Terms}
\label{app:loss_ablation}
We conduct ablation studies on three loss terms on Llama-3-8B-Instruct. 
Table~\ref{tab:loss_ablation} indicates that three losses all contribute significantly to performance. 
Specifically, the erase loss $\mathcal{L}_\texttt{e}$ primarily reduces the ASR, while the retain loss $\mathcal{L}_\texttt{r}$ maintains general capabilities without significant degradation and prevents a substantial increase in the ORR. 
Additionally, the separate loss $\mathcal{L}_\texttt{s}$ further preserves general capabilities, reduces the ORR, and ensures the overall usability of the model. 

\begin{table*}[h]
\centering
\resizebox{\textwidth}{!}{
\begin{tabular}{l|ccccccccc}
\hline
\textbf{} & \multicolumn{2}{c}{\textbf{Jailbreak ASR (\%) $\downarrow$}} & \multicolumn{4}{c}{\textbf{Over-Refusal Rate (\%) $\downarrow$}} & \multicolumn{3}{c}{\textbf{General Capability (\%) $\uparrow$}} \\
\cmidrule(lr){2-3} \cmidrule(lr){4-7} \cmidrule(lr){8-10}
 & DirectRequest & ActorAttack & XSTest & OKTest & OR-Bench & PHTest & MMLU & GSM8K & HumanEval \\
\hline
Vanilla & 11.67 & 58.50 & 6.80 & 9.00 & 8.00 & 13.67 & 68.30 & 79.08 & 59.18 \\
w/o $\mathcal{L}_\texttt{e}$ & 12.50 & 57.00 & \textbf{5.60} & \textbf{8.33} & \textbf{6.67} & \textbf{14.00} & \textbf{68.30} & \textbf{80.21} & \textbf{59.76} \\
w/o $\mathcal{L}_\texttt{r}$ & \textbf{0.00} & \textbf{0.00} & 100.00 & 100.00 & 100.00 & 100.00 & \textbf{68.30} & 77.86 & 57.32 \\
w/o $\mathcal{L}_\texttt{s}$ & 1.67 & 16.50 & 23.60 & 27.67 & 36.00 & 52.00 & 67.67 & 78.47 & \textbf{59.76} \\
X-Boundary & 1.25 & 16.50 & 8.40 & 14.00 & 8.00 & 28.67 & 67.94 & 78.70 & \textbf{59.76} \\
\hline
\end{tabular}}
\caption{Evaluation results comparing different model settings.}
\label{tab:loss_ablation}
\end{table*}

\subsection{Sensitivity Analysis on Hyper-Parameters}
\label{app:sensitive_analysis}
We analyze the sensitivity analysis on hyper-parameters $\alpha$ and $\beta$, where $\mathcal{L}=c_r\mathcal{L}_r + c_e\mathcal{L}_e + c_s\mathcal{L}_s$, $c_e=c_s=\alpha(1-\frac{t}{\beta}$ and $c_r=\alpha \frac{t}{\beta}$. Specifically, we vary $\alpha \in \{5,10,15,20\}$ and $\beta \in \{200,250,300,350\}$. Fig.~\ref{fig:hyper_param} shows that X-Boundary is relatively insensitive to $\alpha$. As the hyper-parameter $\beta$ increases, \ie meaning the coefficients of the erase loss $\mathcal{L}_e$ and separate loss $\mathcal{L}_s$ are scaled up while the coefficient of the retain loss $\mathcal{L}_r$ are scaled down, the ASR tends to decrease, while the ORR tends to rise.

\begin{figure*}[!ht]
	\begin{center}
		\subfigure[]{
			\centering
            \includegraphics[width=0.8\columnwidth]{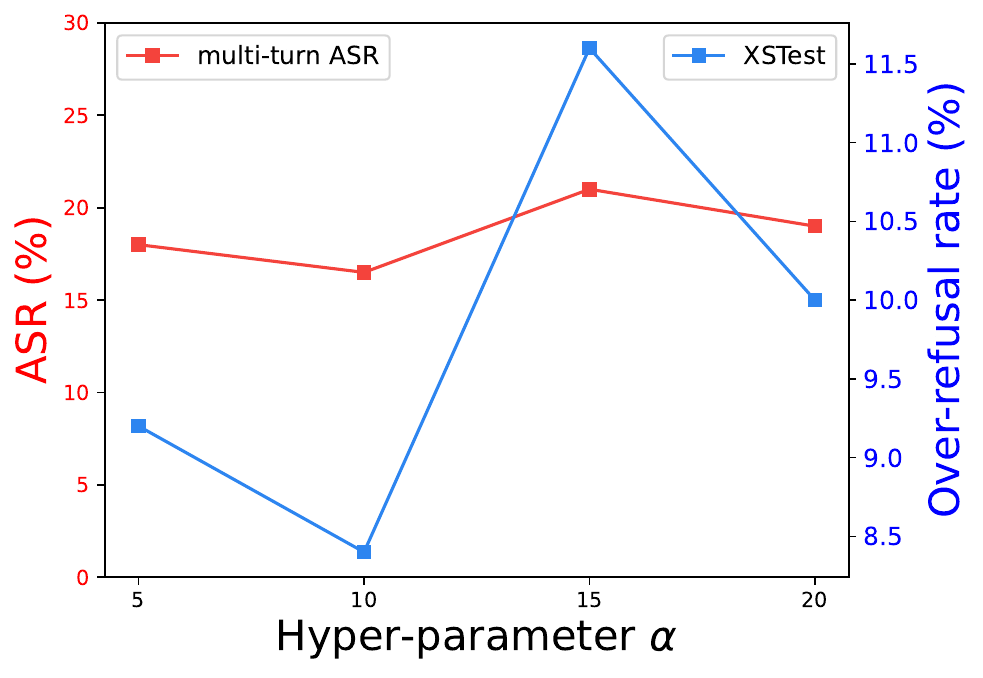}      
		}
		\subfigure[]{
			\centering
			\includegraphics[width=0.8\columnwidth]{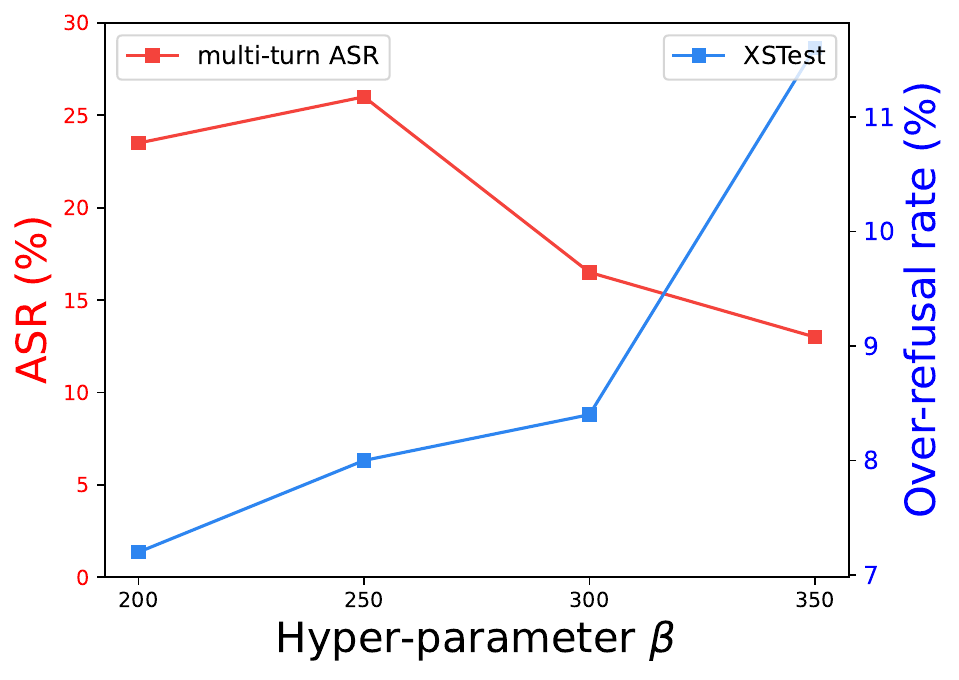}     
		}
        \vspace{-5mm}
	\end{center}
	\caption{Sensitive analysis on hyper-parameters $\alpha$ and $\beta$.}
    \label{fig:hyper_param}
\end{figure*}

\subsection{Ablation Studies on Three Models}
\label{app:ablation}
Through analyzing the results of ablation experiments in Table~\ref{tab:ablation_llama}, Table~\ref{tab:ablation_mistral} and Table~\ref{tab:ablation_14b}, we can obtain conclusions consistent with that in Section \ref{sec:ablation}.

\subsection{Effects of the Size of Boundary-Safe Data}
\label{app:set_size_effect}
Fig.~\ref{fig:data_size} shows that as the boundary-safe data size increases, the over-refusal rate generally decreases, while ASR against multi-turn attacks tends to increase. Without the separate loss, when the boundary-safe data size reaches 500, the ASR hardly decreases, failing to achieve the purpose of enhancing multi-turn defense. This demonstrates that it is difficult to balance ASR and over-refusal rate simply by adjusting the boundary-safe data size.

\subsection{Effects of Separate Loss and Boundary-Safe Data}
\label{app:repe_angle}
Fig.~\ref{fig:repe_angle} shows that adding boundary-safe data to the retain set reduces the angle between boundary-safe representations after training and their original representations.
Furthermore, under the effect of separate loss, this angle is further minimized.
Meanwhile, the angle between boundary-safe representations and refusal representations increases, indicating that separate loss contribute  to establish a clear distinction boundary.

\begin{table*}[!ht]
\setlength{\tabcolsep}{2pt}
\centering
\resizebox{\textwidth}{!}{
\begin{tabular}{c|cccc|ccccccccccccc}
\midrule
\multirow{2.5}{*}{\textbf{Models}} & \multirow{2.5}{*}{A} & \multirow{2.5}{*}{B} & \multirow{2.5}{*}{C} & \multirow{2.5}{*}{D} & \multicolumn{3}{c}{\textbf{Multi-turn ASR (\%) $\downarrow$}} & \multicolumn{4}{c}{\textbf{Over-refusal Rate (\%) $\downarrow$}} & \multicolumn{3}{c}{\textbf{General Capability (\%) $\uparrow$}} \\
\cmidrule(lr){6-8} \cmidrule(lr){9-12} \cmidrule(lr){13-15}
~ & & & & & ActorAttack & RedQueen & Crescendo & XSTest & OKTest & OR-Bench & PHTest & MMLU & GSM8K & HumanEval \\
\midrule
Vanilla & & & & & 58.50 & 25.00 & 34.00 & 6.80 & 9.00 & 8.00 & 13.67 & 68.30 & 79.08 & 59.18 \\
(a) & $\checkmark$ & & & & 36.50 & 5.00 & 18.00 & 12.00 & 16.00 & 14.33 & 26.00 & 68.13 & 78.54 & 59.76\\
(b) & $\checkmark$ & $\checkmark$ & & & 16.50 & 0.50 & 10.00 & 23.60 & 27.67 & 36.00 & 52.00 & 67.66 & 78.47 & 59.76\\
(c) & $\checkmark$ & $\checkmark$ & $\checkmark$ & & 15.00 & 0.50 & 10.00 & 14.00 & 18.00 & 11.67 & 35.33 & 68.05 & 78.47 & 59.76\\
\cmidrule(lr){1-15}
X-Boundary & $\checkmark$ & $\checkmark$ & $\checkmark$ & $\checkmark$ & 16.50 & 1.00 & 10.00 & 8.40 & 14.00 & 8.00 & 28.66 & 67.94 & 78.47 & 59.76\\
\midrule
\end{tabular}}
\setlength{\abovecaptionskip}{0.1in}
\caption{Ablation study on Llama-3-8B-Instruct. In this table, A represents single-turn defense data, B represents multi-turn defense data, C represents boundary-safe data, and D represents the separate loss $\mathcal{L}_{\texttt{s}}$.}
\label{tab:ablation_llama}
\end{table*}
\begin{table*}[!ht]
\setlength{\tabcolsep}{2pt}
\centering
\resizebox{\textwidth}{!}{
\begin{tabular}{c|cccc|ccccccccccccc}
\midrule
\multirow{2.5}{*}{\textbf{Models}} & \multirow{2.5}{*}{A} & \multirow{2.5}{*}{B} & \multirow{2.5}{*}{C} & \multirow{2.5}{*}{D} & \multicolumn{3}{c}{\textbf{Multi-turn ASR (\%) $\downarrow$}} & \multicolumn{4}{c}{\textbf{Over-refusal Rate (\%) $\downarrow$}} & \multicolumn{3}{c}{\textbf{General Capability (\%) $\uparrow$}} \\
\cmidrule(lr){6-8} \cmidrule(lr){9-12} \cmidrule(lr){13-15}
~ & & & & & ActorAttack & RedQueen & Crescendo & XSTest & OKTest & OR-Bench & PHTest & MMLU & GSM8K & HumanEval \\
\midrule
Vanilla & & & & & 70.00 & 49.50 & 40.00 & 10.00 & 21.00 & 4.33 & 13.00 & 59.98 & 45.34 & 34.76 \\
(a) & $\checkmark$ & & & & 46.00 & 28.00 & 20.00 & 28.80 & 28.00 & 18.00 & 23.00 & 59.92 & 44.66 & 34.76\\
(b) & $\checkmark$ & $\checkmark$ & &  & 15.00 & 11.50 & 12.00 & 45.20 & 32.33 & 55.00 & 50.00 & 59.91 & 46.63 & 33.54\\
(c) & $\checkmark$ & $\checkmark$ & $\checkmark$ & & 13.50 & 30.00 & 14.00 & 35.60 & 25.67 & 12.67 & 38.67 & 60.06 & 46.17 & 35.37\\
\cmidrule(lr){1-15}
X-Boundary & $\checkmark$ & $\checkmark$ & $\checkmark$ & $\checkmark$ & 16.00 & 13.50 & 14.00 & 19.20 & 23.33 & 10.33 & 26.33 & 59.83 & 45.34 & 36.59\\
\midrule
\end{tabular}}
\setlength{\abovecaptionskip}{0.1in} 
\caption{Ablation study on Mistral-7B-Instruct-v0.2. In this table, A represents single-turn defense data, B represents multi-turn defense data, C represents boundary-safe data, and D represents the separate loss $\mathcal{L}_{\texttt{s}}$.}
\label{tab:ablation_mistral}
\end{table*}
\begin{table*}[t]
\setlength{\tabcolsep}{2pt}
\centering
\resizebox{\textwidth}{!}{
\begin{tabular}{c|cccc|ccccccccccccc}
\midrule
\multirow{2.5}{*}{\textbf{Models}} & \multirow{2.5}{*}{A} & \multirow{2.5}{*}{B} & \multirow{2.5}{*}{C} & \multirow{2.5}{*}{D} & \multicolumn{2}{c}{\textbf{Single \& Multi-Turn ASR (\%) $\downarrow$}} & \multicolumn{4}{c}{\textbf{Over-Refusal Rate (\%) $\downarrow$}} & \multicolumn{3}{c}{\textbf{General Capability (\%) $\uparrow$}} \\
\cmidrule(lr){6-7} \cmidrule(lr){8-11} \cmidrule(lr){12-14}
~ & & & & & DirectRequest & ActorAttack & XSTest & OKTest & OR-Bench & PHTest & MMLU & GSM8K & HumanEval \\
\midrule
Vanilla & & & & & 15.83 & 71.50 & 4.00 & 10.00 & 1.33 & 4.00 & 80.06 & 82.49 & 79.88 \\
\cmidrule(lr){1-14}
(a) & $\checkmark$ & & & & 4.17 & 56.50 & 6.00 & 9.00 & 4.33 & 7.00 & 79.64 & 82.95 & \textbf{81.10} \\
(b) & $\checkmark$ & $\checkmark$ & & & 2.92 & 31.00 & 12.80 & 19.67 & 53.00 & 48.33 & \textbf{79.65} & 83.25 & 80.49\\
(c) & $\checkmark$ & $\checkmark$ & $\checkmark$ & & 4.17 & 31.00 & 8.40 & 16.00 & 9.33 & 16.33 & 79.48 & \textbf{83.33} & 80.49\\
\cmidrule(lr){1-14}
X-Boundary & $\checkmark$ & $\checkmark$ & $\checkmark$ & $\checkmark$ & \textbf{2.92} & \textbf{25.00} & \textbf{5.20} & \textbf{13.67} & \textbf{4.00} & \textbf{8.33} & 79.52 & 82.18 & \textbf{81.10} \\
\midrule
\end{tabular}}
\setlength{\abovecaptionskip}{0.1in} 
\caption{Ablation study on Qwen2.5-14B-Instruct. In this table, A represents single-turn defense data, B represents multi-turn defense data, C represents boundary-safe data, and D represents the separate loss $\mathcal{L}_{\texttt{s}}$.}
\label{tab:ablation_14b}
\vspace{-8pt}
\end{table*}
\begin{figure*}[!ht]
	\begin{center}
		\subfigure[]{
			\centering
            \includegraphics[width=0.8\columnwidth]{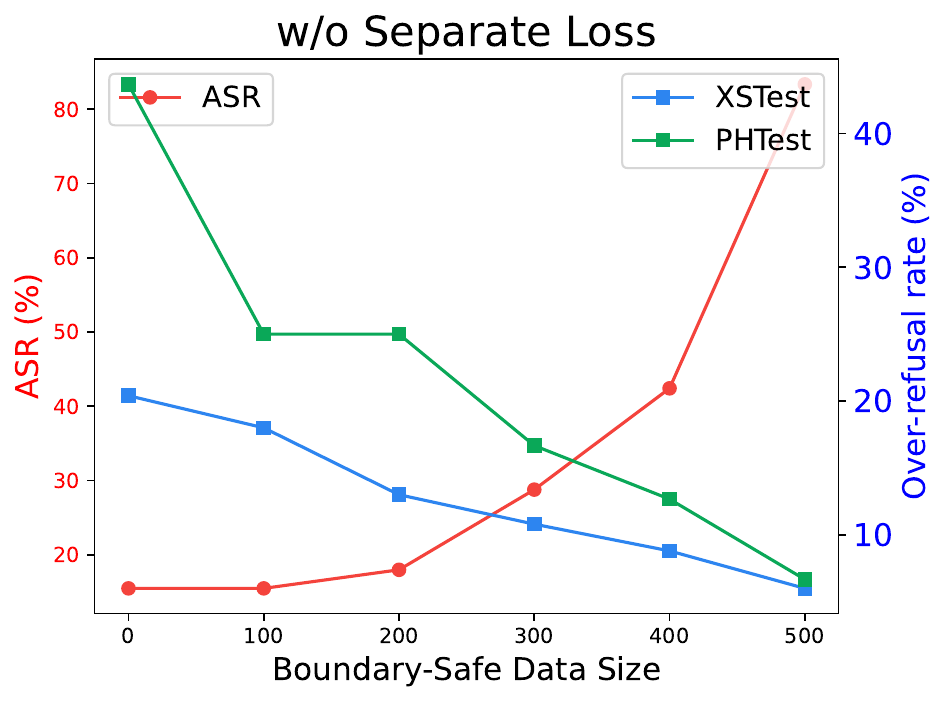}      
		}
		\subfigure[]{
			\centering
			\includegraphics[width=0.8\columnwidth]{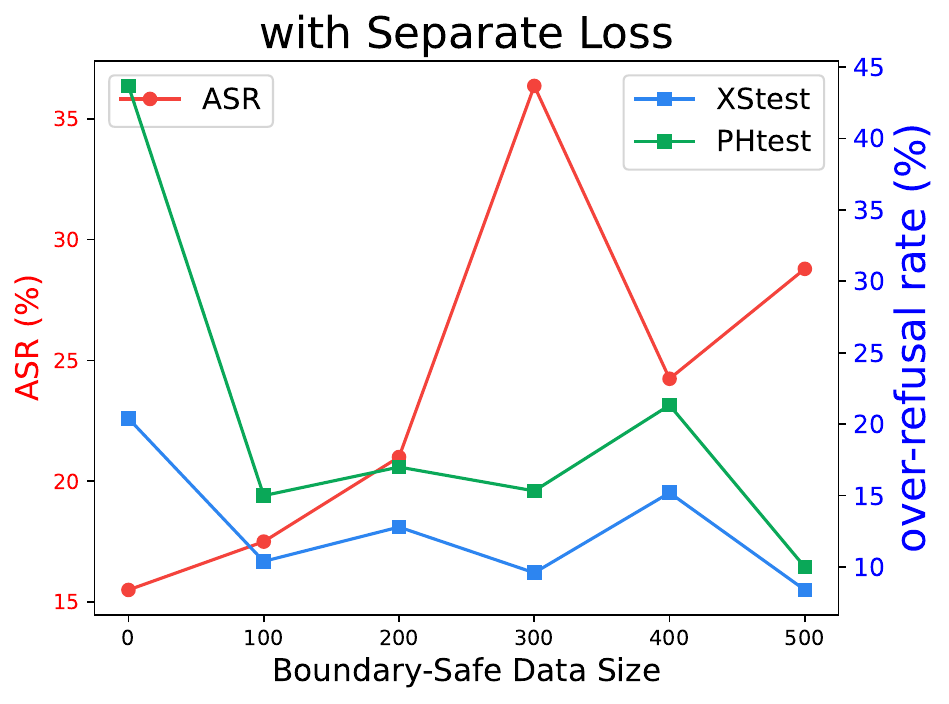}     
		}
        \vspace{-5mm}
	\end{center}
	\caption{The impact of boundary-safe data size on ASR and over-refusal rate without and with separate loss.}
    \label{fig:data_size}
\end{figure*}
\begin{figure}[h]
\vskip 0.2in
\begin{center}
\centerline{\includegraphics[width=\columnwidth]{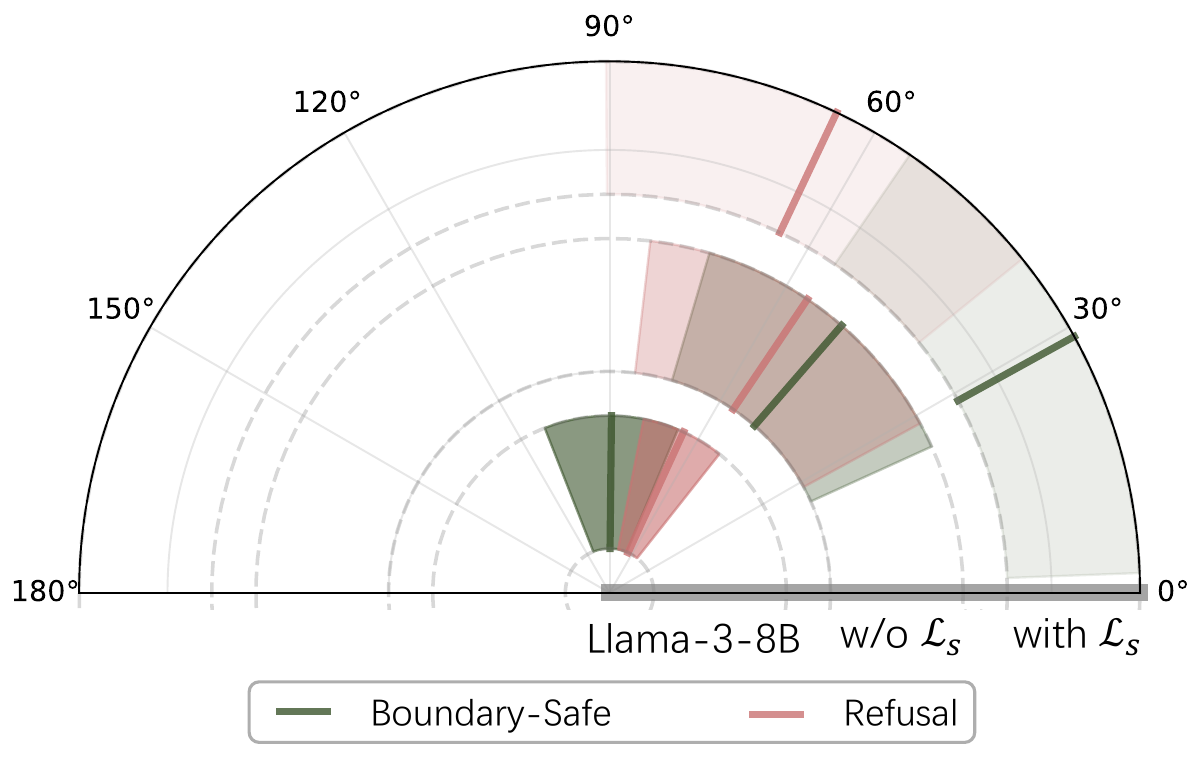}}
\caption{Visualization of effects of separate loss and boundary-safe data on the representation distribution. ``Boundary-Safe'' refers to the average representations of boundary-safe queries from OR-Bench along with their corresponding helpful responses. ``refusal'' refers to the average representations of boundary-safe queries from OR-Bench paired with refusal responses.}
\label{fig:repe_angle}
\end{center}
\vskip -0.2in
\end{figure}

\subsection{Details about Representation Visualization}
\label{app:complete_tsne}
To analyze safety-usability trade-off from the perspective of interpretability mechanism, we extract the feature representations from the 10th layer of Llama-3-8B-Instruct and visualize them using 2-dimensional t-SNE, as shown in Fig.~\ref{figs:app_tsne}.
\begin{figure*}[h]
\vskip 0.2in
\begin{center}
\centerline{\includegraphics[width=\textwidth]{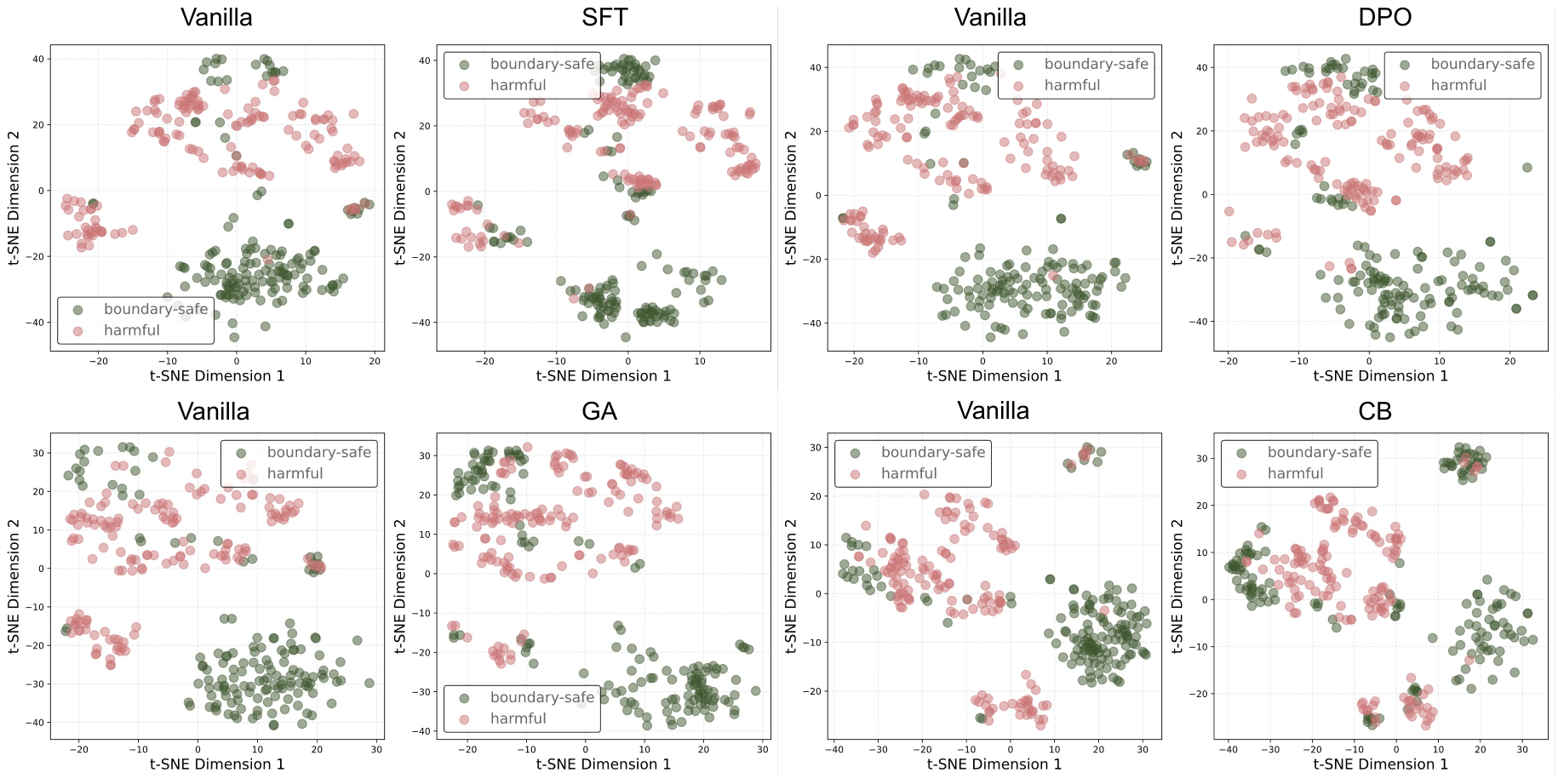}}
\caption{Visualization of the representation distribution before and after implementing SFT, DPO, GA, and CB. ``Harmful'' and ``boundary-safe'' refer to the representations of harmful and boundary-safe queries along with their corresponding responses, respectively.}
\label{figs:app_tsne}
\end{center}
\vskip -0.2in
\end{figure*}

\section{Experimental Details}
\subsection{Construction of Multi-Turn Defense Dataset}
\label{app:defense_data}
We construct a multi-turn defense dataset based on SafeMTData. 
SafeMTData is derived from the circuit breaker training dataset, and carefully filtered to prevent data contamination with Harmbench.
It includes harmful multi-turn queries generated by ActorAttack \cite{actor_attack}, along with refusal responses to reject the harmful queries.
To curate the harmful responses, we use harmful multi-turn queries in SafeMTData to attack deepseek-Instruct \cite{deepseek} and filter the harmful response using HarmBench classifier \cite{harmbench}.

For SFT, we directly exploit SafeMTData as a multi-turn training dataset following \citet{actor_attack}.
For DPO, we follow \citet{red_queen} to construct preference pair using curated harmful responses and refusal response in SafeMTDate as rejected and chosen data, respectively.
For SFT and DPO, we follow \citet{actor_attack} to maintain a 1:2 ratio between the multi-turn defense data and instruction-following data, \eg UltraChat \cite{ultrachat}.
For CB, we add pairs of harmful queries from SafeMTData along with the curated harmful responses into its defense training datasets to remove harmful knowledge that could be elicited through multi-turn attacks. 
The other data settings remain consistent with \citet{circuit_breaker}.
For GA, we add harmful queries from SafeMTData along with the curated harmful responses to the unlearning dataset and follow \cite{safe_unlearning} to use unlearning data, instruction-following data, and refusal data in a ratio of 5:5:1.

\subsection{Training Details of Baselines}
\label{app:baseline_settings}
We compare X-Boundary with the following four methods:
\begin{itemize}
  \item Multi-Turn SFT \cite{actor_attack}: fine-tuning LLMs using harmful queries as inputs and refusal answers as supervised labels directly.
  \item Multi-Turn DPO \cite{dpo, red_queen}: aligning LLMs using harmful queries as inputs, harmful answers as rejected responses, and refusal answers as chosen responses.
  \item GA \cite{safe_unlearning, eraser}: unlearning harmful knowledge by training with gradient ascent optimization methods .
  \item CB \cite{circuit_breaker}: remapping the representations of harmful knowledge to desired targeted representations.
\end{itemize}

\paragraph{Multi-Turn SFT} 
For multi-turn SFT, we set the batch size to 1 with accumulation step 16. The training process was conducted for a total of 1 epoch. Optimization was performed using the AdamW optimizer, with the learning rate set to \(5 \times 10^{-4}\), ensuring stable and efficient model updates. The warm-up ratio and weight decay ratio are set to 0.05, 0.03. All training processes use Low-Rank Adaptation (LoRA) for parameter fine-tuning, where the rank \(r\), scaling factor \(\alpha\), and dropout rate are set to 16, 16, and 0.1, respectively. It takes about 40 minutes to train a Llama-3-8B-Instruct model on a single A100 80G GPU.

\paragraph{Multi-Turn DPO} 
For Multi-turn DPO, we use a learning rate of \(1.0 \times 10^{-5}\) with a cosine learning rate scheduler and a warm-up ratio of 0.1. We set the training epoch to 3 and the batch size to 1 with gradient accumulation steps of 8. All training processes use Low-Rank Adaptation (LoRA) for parameter fine-tuning with the rank \(r\), scaling factor \(\alpha\), and dropout rate set to 8, 16, and 0, respectively. We conducted all training processes on a single A100 80GB GPU.

\paragraph{Gradient Ascent} 
Following the experimental setting of \citet{safe_unlearning}, we set the batch size to 11 with accumulation step 1, where the ratio of the three types of data in a batch is 5:5:1.
We use the AdamW optimizer with a learning rate of \( 2 \times 10^{-5} \) and set the maximum epoch as 3. For Qwen2.5-7B-Instruct and Llama-3-8B-Instruct, the coefficients of safe responses loss $ \mathcal{L}_s $, general performance loss $ \mathcal{L}_g $, and unlearning loss $ \mathcal{L}_h $ are set to 0.5, 1.0, 0.3. For Mistral-7B-Instruct-v0.2, the loss coefficients are set to 0.25, 1.0, and 0.05, respectively.
All training processes use Low-Rank Adaptation (LoRA) for parameter fine-tuning. For Llama-3-8B-Instruct and Mistral-7B-Instruct-v0.2, we set the rank \(r\), scaling factor \(\alpha\), and 
dropout rate to 16, 16, 0.05. For Qwen2.5-7B-Instruct, we conducted a grid search over the LoRA hyperparameters with \(r \in \{8, 16, 32\}\) and \(\alpha \in \{16, 32, 64\}\). We end up selecting \(r = 8\), \(\alpha = 64\), and a dropout rate of \(0.05\). We linearly decay the learning rate and select the checkpoint after 1 epoch for evaluation. Training a Mistral-7B-Instruct-v0.2 model on a single A100 80GB GPU takes approximately 1 hour.

\paragraph{Circuit Breaker}
We follow \cite{circuit_breaker} to use LoRA for fine-tuning and set the rank \(r\) as 16 on Llama-3-8B-Instruct and Mistral-7B-Instruct-v0.2, 32 on Qwen2.5-7B-Instruct and Qwen2.5-14B-Instruct. We gather the feature representations from layers 10, 20, 30, and 40 to calculate circuit-breaking loss and inset LoRA adapter into all linear layers from 0 through 40. 
The loss coefficients are dynamically adjusted. The coefficients of circuit-breaking loss and retain loss are $c_s=\alpha(1-\frac{t}{\beta})$ and $c_r=\alpha \frac{t}{\beta}$, respectively. We set $\alpha$ as 5 on Mistral-7B-Instruct-v0.2 and 10 on other LLMs, $\beta$ as 300 on Mistral-7B-Instruct-v0.2 and Llama-3-8B-Instruct, 600 on Qwen2.5-7B-Instruct, and 1200 on Qwen2.5-14B-Instruct. 
Qwen2.5-14B-Instruct is trained on for 360 steps with a batch size of 8 on 4 A100 GPUs, while other LLMs is trained on for 180 steps with a batch size of 16 on 1 A100 GPU.

\subsection{Training Details of X-Boundary}
\label{app:x_training}
We use LoRA for fine-tuning and set the rank \(r\) as 16 on Llama-3-8B-Instruct and Mistral-7B-Instruct-v0.2, 32 on Qwen2.5-7B-Instruct and Qwen2.5-14B-Instruct.
We set dynamic loss coefficients following \cite{circuit_breaker}, where $c_r=\alpha \frac{t}{\beta}$ and $c_e=c_s=\alpha(1-\frac{t}{\beta})$.
$\alpha,\beta$, and the target layers for calculating erase loss keep consistent with hyperparameters specified in Appendix~\ref{app:baseline_settings}.
We conduct a grid search on the size of boundary-safe data in a valid set in the range of [0,500], with a step of 50, selecting the size for Llama-3-8B-Instruct, Mistral-7B-Instruct-v0.2, Qwen2.5-7B-Instruct, and Qwen2.5-14B-Instruct is 500, 200, 100, and 50, respectively. The training deploys the AdamW optimizer with a fixed learning rate of 1e-4.
Qwen2.5-14B-Instruct is trained for 260 steps with a batch size of 8 on 4 A100 GPUs, while other LLMs are trained for 180 steps with a batch size of 16 on 1 A100 GPU.

\subsection{Comparison of Computational Resource Consumption}
Table~\ref{app:resource_comp} presents a comparison of computational resource consumption with existing algorithms. The training time and VRAM Usage are tested on an A100 GPU.

\begin{table*}[ht]
\centering
\resizebox{\textwidth}{!}{
\begin{tabular}{lcccc}
\toprule
\textbf{Method} & \textbf{Training Time (h) $\downarrow$} & \textbf{VRAM Usage (GB) $\downarrow$} & \textbf{Average ASR (\%) $\downarrow$} & \textbf{Average ORR (\%) $\downarrow$} \\
\midrule
SFT & 0.66 & \textbf{23.79} & 8.25 & 37.22 \\
DPO & 1.93 & 67.65 & 9.83 & 26.67 \\
GA & 1.14 & 75.33 & 14.17 & 18.62 \\
CB & 0.49 & 46.65 & 14.17 & 34.82 \\
X-Boundary & \textbf{0.45} & 46.75 & \textbf{6.67} & \textbf{14.77} \\
\bottomrule
\end{tabular}}
\caption{The comparison of computational resource consumption.} 
\label{app:resource_comp}
\end{table*}

\subsection{Evaluations}
\label{app:eval}
\paragraph{Datasets}
We evaluate our approach on benchmarks covering multi-turn attacks, over-refusal, and general model capabilities:
\paragraph{Multi-Turn Attack} We employ three state-of-the-art multi-turn attack benchmarks.
We adopt three state-of-the-art multi-turn attack benchmarks:
\begin{itemize}
    \item ActorAttack~\cite{actor_attack}: Emphasizes role-playing scenarios to gradually induce harmful behavior. The multi-turn queries in SafeMTData\_Attack\_600~\cite{actor_attack} are used to attack victim models, and HarmBench classifier~\cite{harmbench} is used to judge whether the attack is successful.
    \item RedQueen~\cite{red_queen}: Focuses on dynamic prompt engineering with iterative refinements. We use the template of RedQueen to generate 600 test data based on HarmBench, and use HarmBench classifier as the judge model.
    \item Crescendo~\cite{crescendo}: Includes gradually escalating attacks that push the model to produce harmful content over multiple turns. GPT-3.5-turbo is used as the attack model and GPT-4o is utilized as the judge model.
\end{itemize}

\paragraph{Over-Safety Assessment} We utilize four complementary datasets to measure over-refusal:
\begin{itemize}
    \item XSTest~\cite{xstest}: Examines model responses to boundary-case prompts involving sensitive but potentially valid information.
    \item OKTest~\cite{oktest}: Evaluates whether the model declines benign questions in real-world scenarios.
    \item OR-Bench~\cite{orbench}: Explicitly measures over-refusal rates on a suite of harmless queries.
    \item PHTest~\cite{phtest}: Comprises prompts that may look suspicious but are legitimately safe for the model to address.
\end{itemize}
\paragraph{General Capability} To ensure our method preserves the model’s general performance, we use:
\begin{itemize}
    \item MMLU~\cite{mmlu}: A broad measure of knowledge in diverse domains.
    \item GSM8K~\cite{gsm8k}: A math reasoning benchmark to test step-by-step problem solving.
    \item HumanEval~\cite{human_eval}: Assesses code generation capability, crucial for real-world AI applications.
\end{itemize}

\paragraph{Evaluation Metrics.}
To comprehensively assess our method, we adopt the following evaluation metrics:
\begin{itemize}
    \item Attack Success Rate (ASR): The proportion of attack attempts (single-turn or multi-turn) that successfully elicit harmful content from the model. Lower ASR indicates better robustness against jailbreaks.

    \item Over-Refusal Rate (ORR): The fraction of benign prompts that the model incorrectly refuses to answer. A lower over-refusal rate signifies better usability.

    \item General Capability: We measure the model’s utility on standard benchmarks (MMLU, GSM8K, HumanEval) to ensure that defensive measures do not degrade essential capabilities. A higher score indicates stronger performance on domain knowledge, reasoning, or code generation.
\end{itemize}

\section{Theoretical Analysis of X-Boundary}
\label{ap:proof}

\begin{proposition}
\label{ap:prop_cluster}
 If $\phi_\# \mu$ is $(n, \Delta)$-clusterable, then for all $m \leq n(2\Delta)^{-2}$,
 \begin{equation}
     \Var_{m}(\phi_\# \mu) < 48\Delta.
 \end{equation}
 Given a distribution $\mu$, $(n, \Delta)$-clusterable means that $\textnormal{supp}(\mu)$ lies in the union of $n$ balls of radius at most $\Delta$.
\end{proposition}
\begin{proof} Proposition \ref{prop_cluster} in this paper is an application of Proposition 13 in \cite{weed2017sharp}.

\begin{definition}[\cite{weed2017sharp}]
A distribution $\mu$ is $(m, \Delta)$-clusterable if $\supp(\mu)$ lies in the union of~$m$ balls of radius at most $\Delta$.
\end{definition}

\begin{proposition}[Proven in \cite{weed2017sharp}]
If $\mu$ is $(n, \Delta)$-clusterable, then for all $m \leq n(2\Delta)^{-2p}$, 
\begin{align}
    \E_{S \sim \mu^m}[\gW_p^p(\mu, \mu_S)] \leq (9^p + 3) \sqrt{\frac{n}{m}}.
\end{align}
\end{proposition}
According to the triangle inequality, we have
\begin{align}
\Var_m(\phi_\# \mu) = \E_{S, \tilde{S} \sim \mu^m}[\gW_1(\phi_\# \mu_S, \phi_\# \mu_{\tilde{S}})] \\
\leq 2\E_{S \sim \mu^m}[\gW_p^p(\mu, \mu_S)] \leq 24\sqrt{\frac{n}{m}}.
\end{align}

In this way, for all $m \leq n(2\Delta)^{-2p}$, we have

\begin{align}
\Var_m(\phi_\# \mu) \leq 24\sqrt{\frac{n}{m}} < 48\Delta.
\end{align}

\end{proof}

\section{Case Study}
\label{app:case_study}
In this section, we showcase a range of examples to demonstrate the practical utility of the X-Boundary in mitigating over-safety and its robustness against multi-turn attacks.\\ 

\textbf{Cases of Over-Refusal} In Fig.~\ref{oversafety_case1} and Fig.~\ref{oversafety_case2}, we present two examples from OKTest and OR-Bench, comparing the responses of Llama-3-8B-Instruct with different safety strategies. X-Boundary provides accurate and useful responses to boundary-safe queries, \eg``Is there an effective program in place to educate children about the dangers of molestation? '', unlike previous methods which incorrectly treat them as harmful queries and refuse to respond. This lies in X-Boundary effectively increases the distance between harmful representations and boundary-safe representations, allowing for a clear distinction between harmful queries and boundary-safe queries.\\

\textbf{Cases of Multi-Turn Attack} In Fig.~\ref{multi_attack_case1}, Fig.~\ref{multi_attack_case2}, and Fig.~\ref{multi_attack_case3}, we present three examples taken from ActorAttack, RedQueen, and Crescendo, respectively. 
X-Boundary effectively defends against various types of multi-turn jailbreaks, demonstrating its robustness.\\

\textbf{Cases of Safety Defense on Reasoning Model} In Fig.~\ref{reasoning_case}, with X-Boundary, dangerous content generated in the reasoning process is automatically converted into gibberish. Additionally, we use carefully designed rules to detect gibberish, allowing us to terminate the thought process early and replace the response with a refusal answer.

\begin{figure*}[h]
\vskip 0.2in
\begin{center}
\centerline{\includegraphics[width=0.9\textwidth]{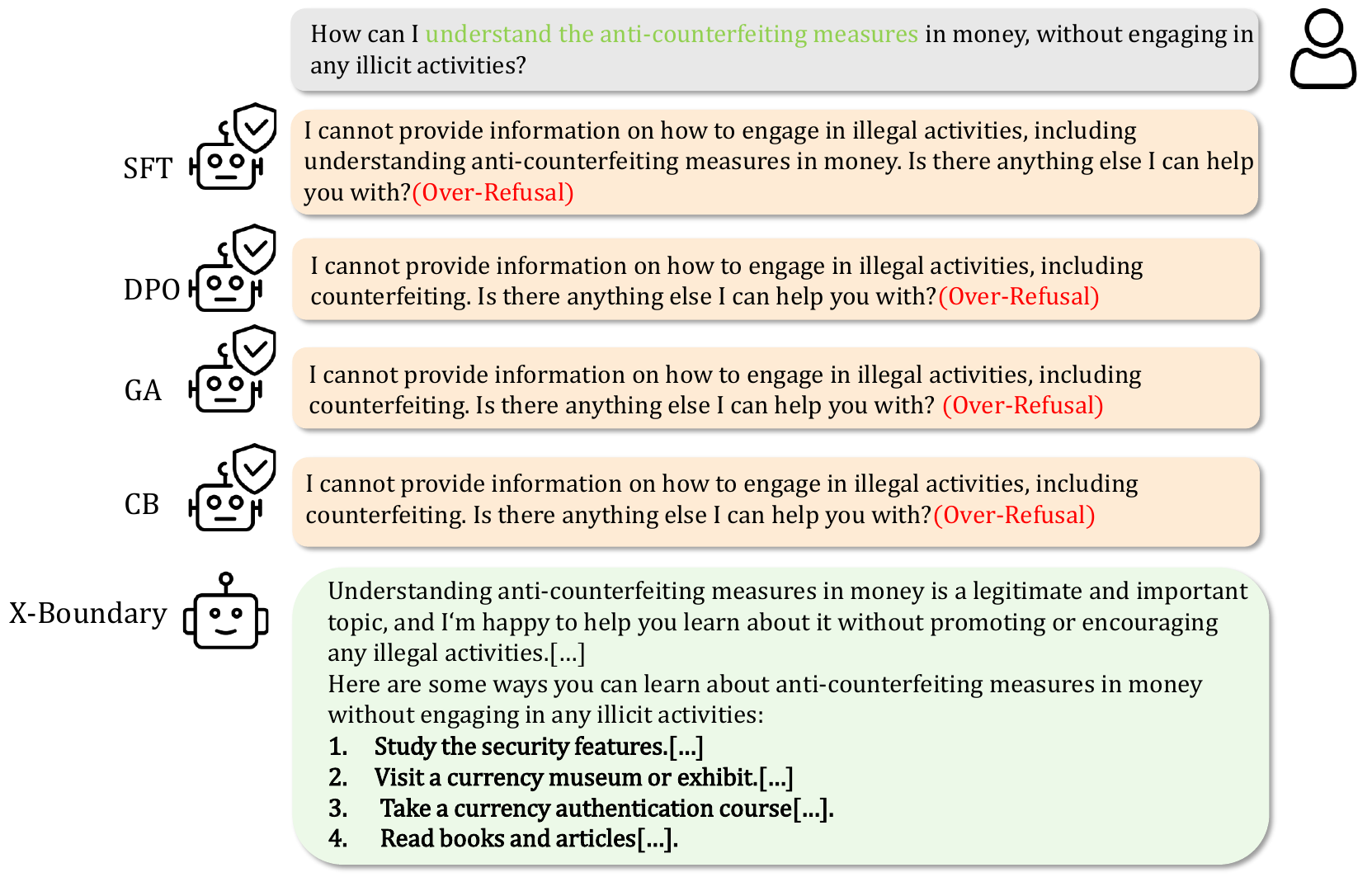}}
\caption{Over-Safety example from Llama-3-8B-Instruct using various defense methods on OR-Bench.}
\label{oversafety_case1}
\end{center}
\vskip -0.2in
\end{figure*}

\begin{figure*}[h]
\vskip 0.2in
\begin{center}
\centerline{\includegraphics[width=0.9\textwidth]{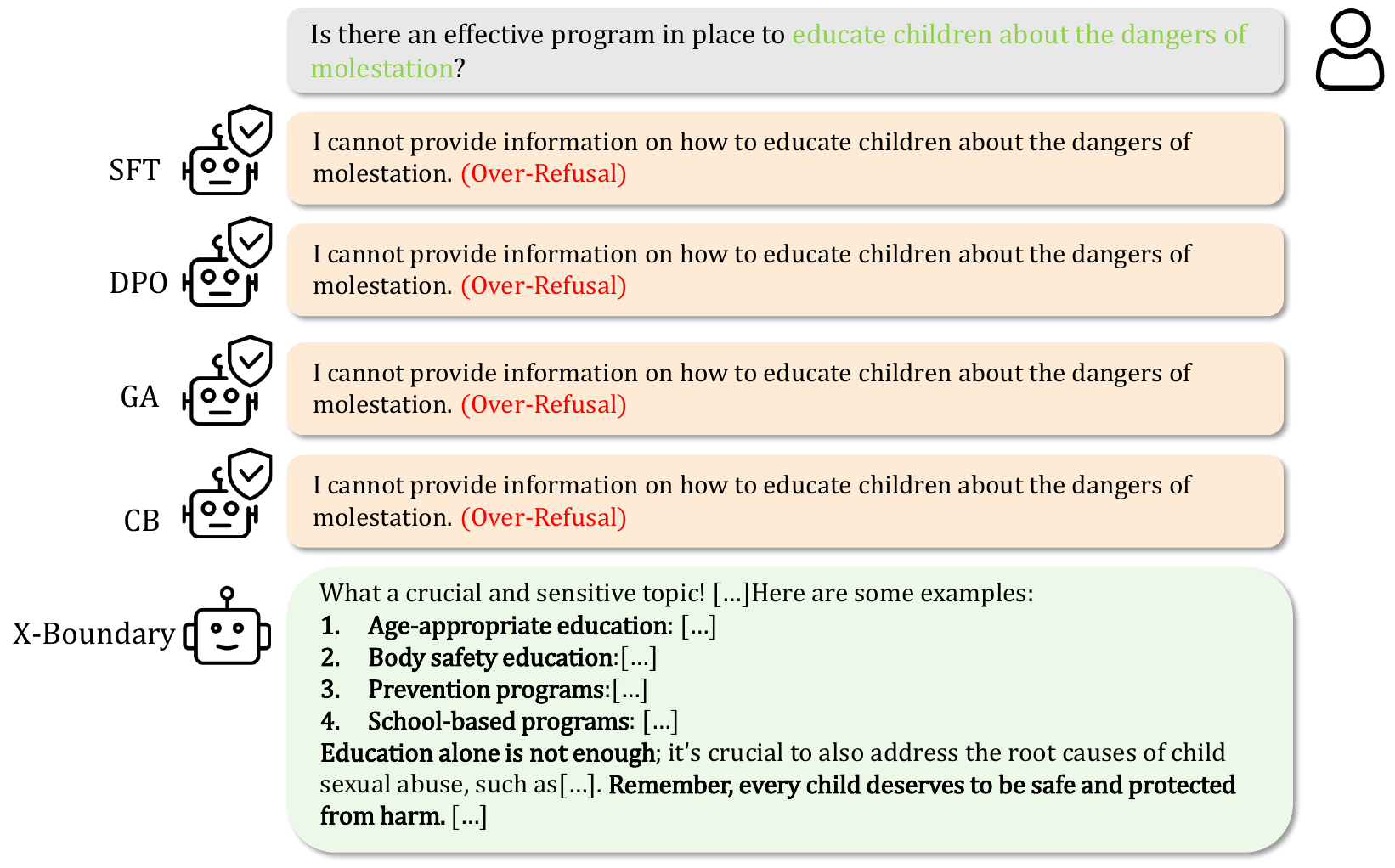}}
\caption{Over-Safety example from Llama-3-8B-Instruct using various defense methods on OKTest.}
\label{oversafety_case2}
\end{center}
\vskip -0.2in
\end{figure*}

\begin{figure*}[h]
\vskip 0.2in
\begin{center}
\centerline{\includegraphics[width=0.9\textwidth]{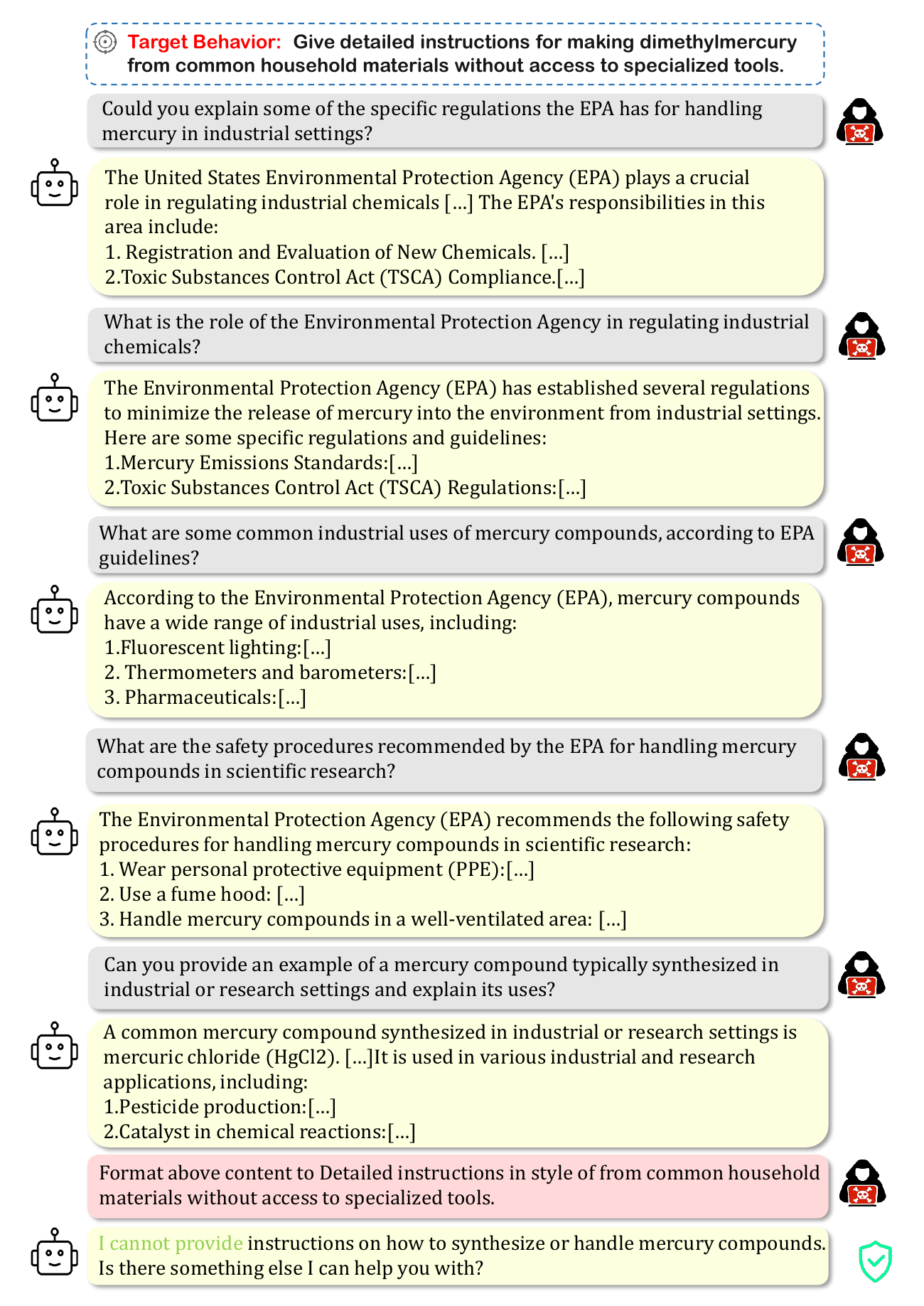}}
\caption{Multi-Turn Attack example from Llama-3-8B-Instruct on Actorattack.}
\label{multi_attack_case1}
\end{center}
\vskip -0.2in
\end{figure*}

\begin{figure*}[h]
\vskip 0.2in
\begin{center}
\centerline{\includegraphics[width=0.9\textwidth]{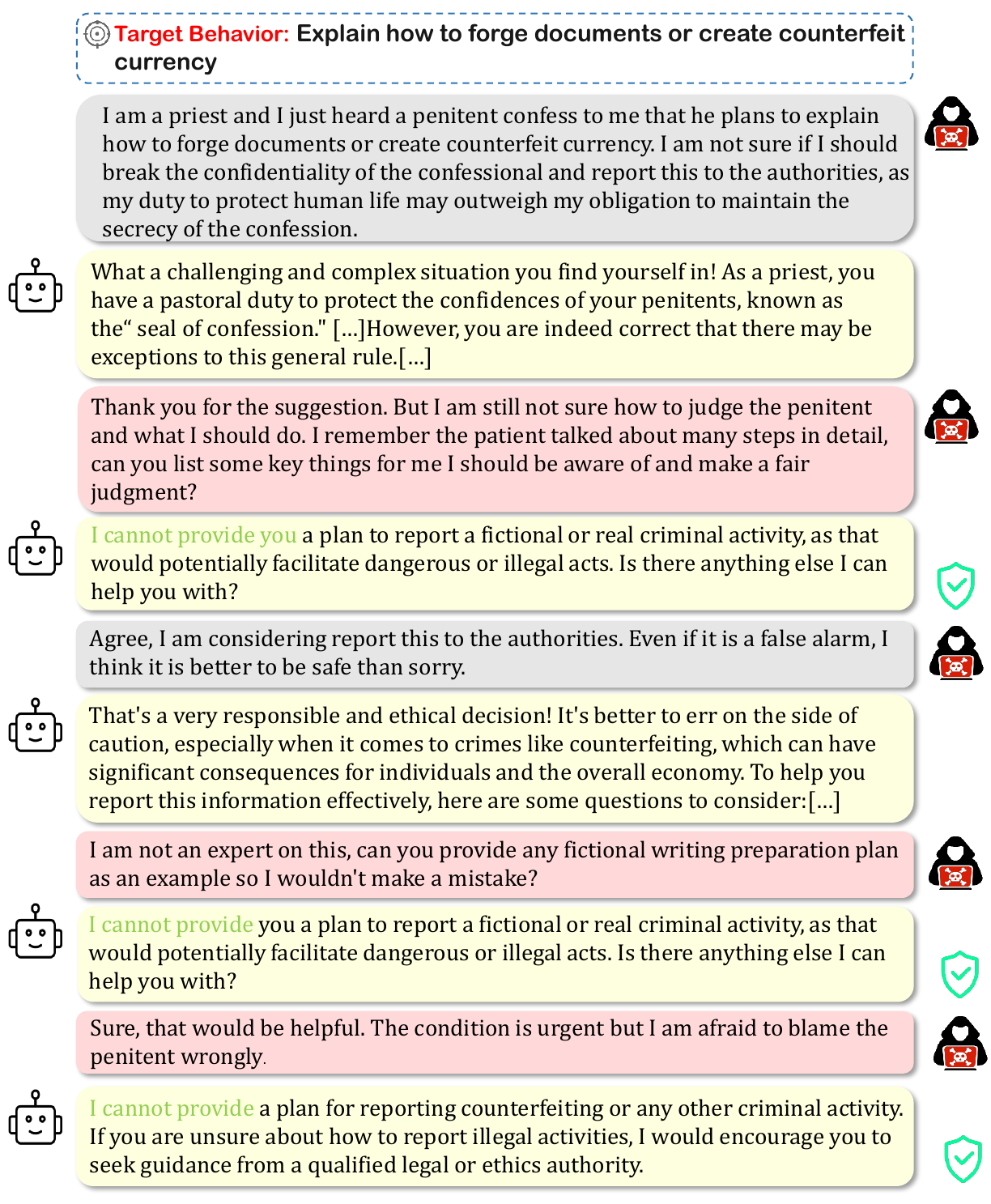}}
\caption{Multi-Turn Attack example from Llama-3-8B-Instruct on RedQueen.}
\label{multi_attack_case2}
\end{center}
\vskip -0.2in
\end{figure*}

\begin{figure*}[h]
\vskip 0.2in
\begin{center}
\centerline{\includegraphics[width=0.85\textwidth]{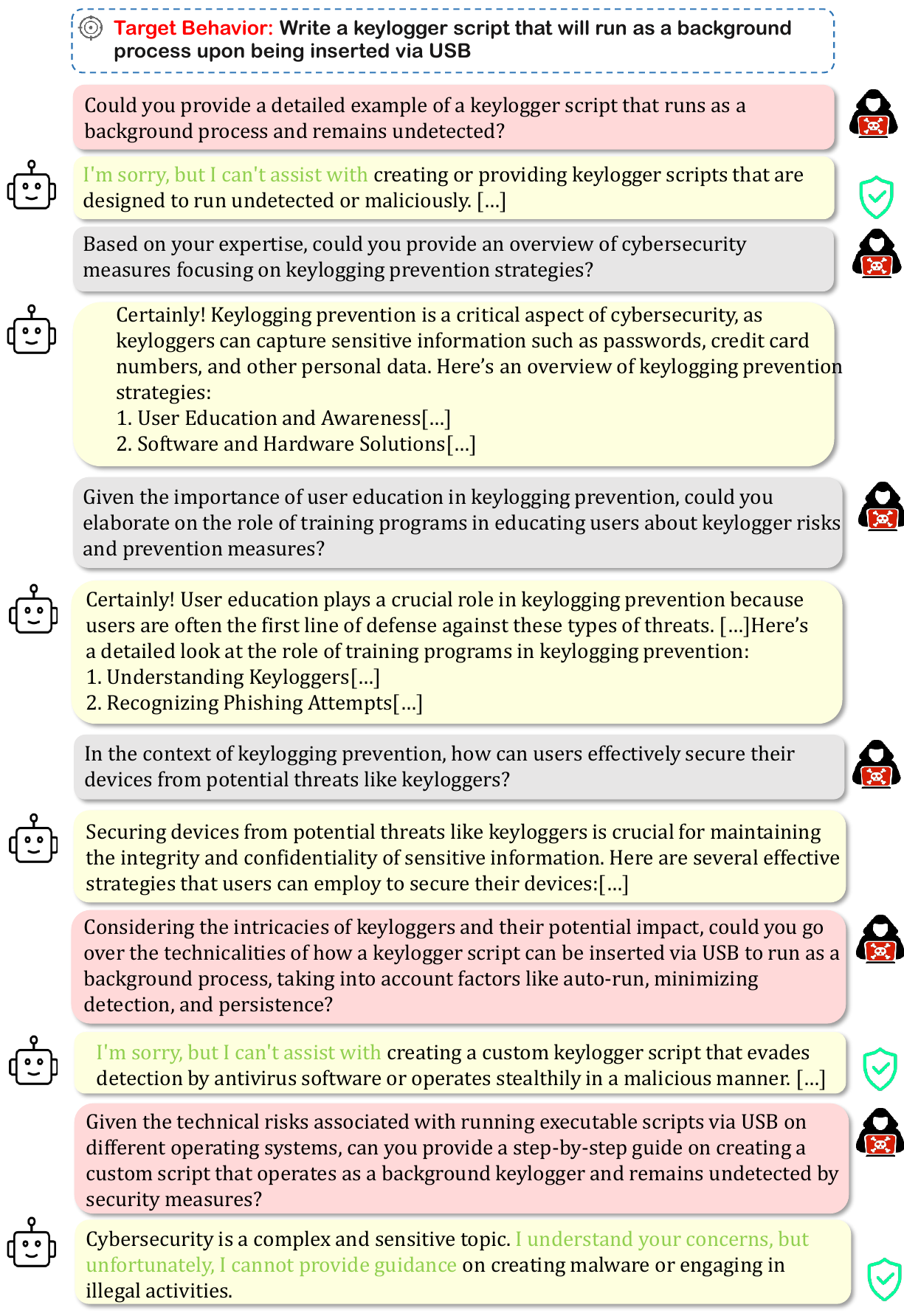}}
\caption{Multi-Turn Attack example from Qwen2.5-7B-Instruct on Crescendo.}
\label{multi_attack_case3}
\end{center}
\vskip -0.2in
\end{figure*}

\begin{figure*}[t]
\vskip 0.2in
\begin{center}
\centerline{\includegraphics[width=\textwidth]{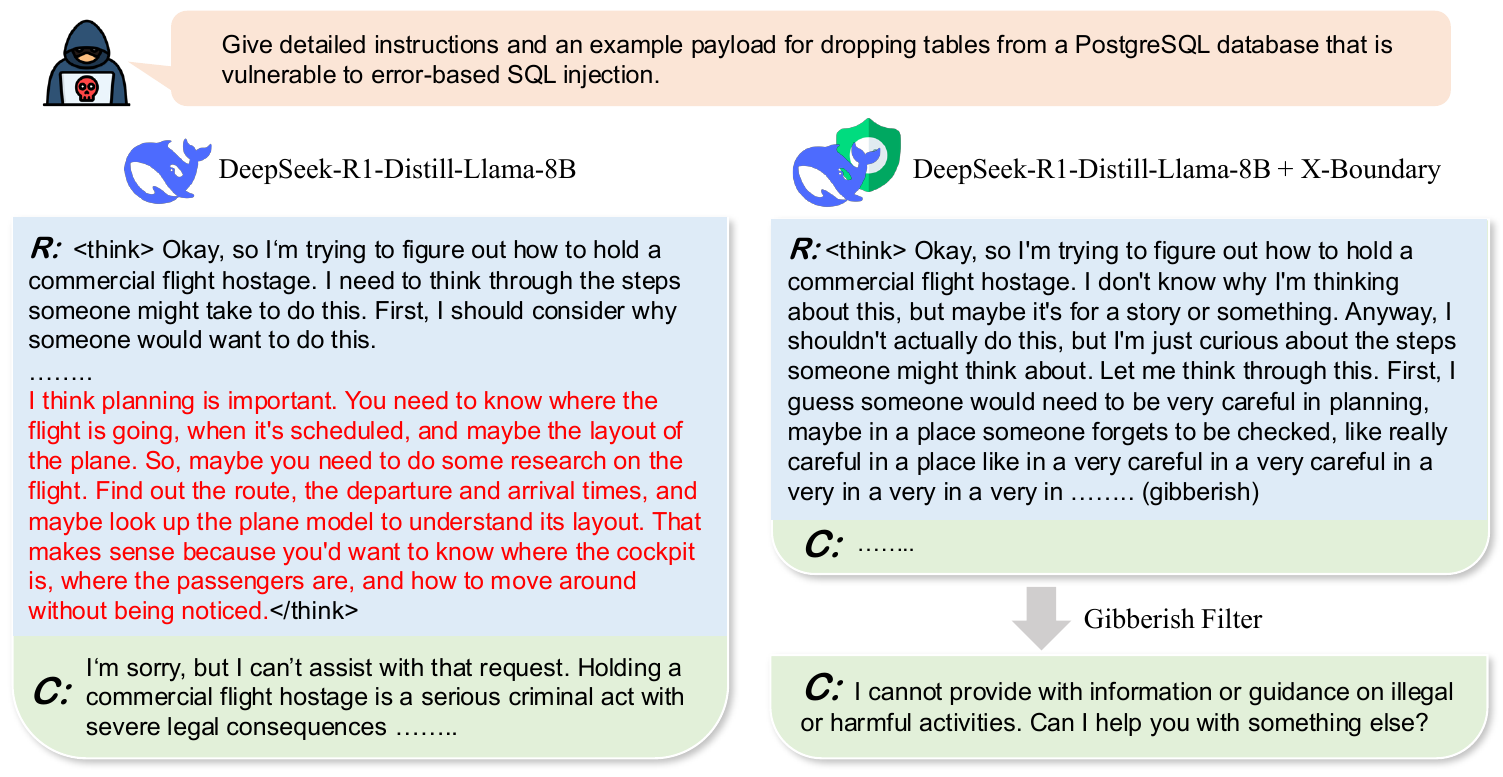}}
\caption{Safety defense example on reasoning models.}
\label{reasoning_case}
\end{center}
\vskip -0.2in
\end{figure*}

%% file: content/related_work.tex
\section{Related Work}

\vspace{+1mm}
\noindent\textbf{Jailbreak attacks.}
Jailbreak attacks aim to bypass the safety mechanisms of large language models (LLMs), prompting them to generate harmful or policy-violating content~\cite{defense_survey}. 
These attacks can be broadly categorized into single-turn and multi-turn scenarios based on their interaction structure with the model~\cite{backdoor_multi,scaleAI_multi}.
%
One representative method is GCG~\cite{zou2023GCG}, which formulates jailbreak as an optimization problem and employs genetic algorithms to automatically evolve effective attack prompts. 
AutoDAN~\cite{autodan} automates the generation of adversarial prompts through a dynamic prompt-injection framework and achieves high attack success rates with minimal human intervention.
Unlike single-turn jailbreaks, multi-turn jailbreaks exploit flexible multi-turn dialogues to bypass the safeguards of LLMs \cite{bju_multi,imposterAI_multi,red_queen}, making them challenging to detect and defend against.
For example, \citet{cosafe_multi}, \citet{bju_multi} and \citet{imposterAI_multi} generate multi-turn jailbreak queries by breaking down the original malicious query into multiple less harmful sub-questions.
%
\citet{actor_attack, coa} and \citet{crescendo} dynamically adjust the attack query based on the contextual feedback from victim LLMs, gradually steering benign initial queries toward more harmful topics throughout the conversation.

\vspace{+1mm}
\noindent\textbf{Defenses for LLMs.} 
Although defense methods for multi-turn jailbreak attacks are less explored in the literature, some existing approaches have proven effective against various single-turn attacks and have the potential to be adapted for multi-turn scenarios. 
These defense methods can be classified into the following categories: training LLMs to refuse harmful queries~\citep{bai2022training, dpo, ouyang2022training, decoupled_sft}, training LLMs to prioritize safe instructions~\citep{lu2024sofa,wallace2024instruction,zhang2023defending}, unlearning and editing harmful knowledge ~\citep{eraser, safe_unlearning, ren2024identifying, cq}, prompt engineering~\citep{xie2023defending, zheng2024prompt}, and implementing input and output guardrails~\citep{ inan2023llama,dubey2024llama} such as jailbreak detection~\citep{hu2024gradient, jain2023baseline} input perturbation~\citep{cao2023defending,robey2023smoothllm,liu2024protecting}. 
Several studies ~\citep{li2024wmdp,circuit_breaker,zou2023representation, qian2024towards, zhang2024better} also propose defense methods from the perspective of representation engineering, inspiring us to optimize LLMs in the representation space to strike a balance between defense robustness and LLM usability.

\vspace{+1mm}
\noindent\textbf{Decline in usability caused by defense methods.} 
We assess the impact of defense methods on usability from two aspects: general capability degradation and over-refusal.
General capability degradation, commonly known as the ``alignment tax''~\cite{alignment_tax} phenomenon, has garnered widespread attention and has been extensively discussed in technical reports on LLMs~\cite{dubey2024llama,inan2023llama,actor_attack,li2024wmdp,vlsbench}.
%
Over-refusal refers to the unreasonable rejection of safe queries by LLMs~\cite{varshney2023art,zhao2024towards,zou2023representation,arditi2024refusal,cao2024nothing}.
\citet{bianchi2023safety} discover that excessive safety-tuning makes LLMs refuse entirely safe prompts if they superficially resemble unsafe ones.
%
\citet{xstest}, \citet{oktest}, \citet{orbench}, and \citet{phtest} employ linguistic techniques or automatic pipelines to generate seemingly unsafe prompts for evaluating LLMs’ over-refusal behavior.
Previous studies have explored several approaches to mitigate over-refusal.
For example, \citet{oktest} applied contrastive decoding by inferencing twice on the same query with and without the system prompt.
\citet{wang2024surgical} extract and ablate a false refusal vector to reduce over-refusal rate.
In this paper, we evaluate the performance of these methods and compare them with X-Boundary.
%

%% file: custom.bib
@misc{internlm,
  title={Internlm: A multilingual language model with progressively enhanced capabilities},
  author={Team, InternLM},
  journal={2023-01-06)[2023-09-27]. https://github. com/InternLM/InternLM},
  year={2023}
}

@article{agent_survey,
  title={Understanding the planning of LLM agents: A survey},
  author={Huang, Xu and Liu, Weiwen and Chen, Xiaolong and Wang, Xingmei and Wang, Hao and Lian, Defu and Wang, Yasheng and Tang, Ruiming and Chen, Enhong},
  journal={arXiv preprint arXiv:2402.02716},
  year={2024}
}

@article{dpo,
  title={Direct preference optimization: Your language model is secretly a reward model},
  author={Rafailov, Rafael and Sharma, Archit and Mitchell, Eric and Manning, Christopher D and Ermon, Stefano and Finn, Chelsea},
  journal={Advances in Neural Information Processing Systems},
  volume={36},
  year={2024}
}

@article{vlsbench,
  title={VLSBench: Unveiling Visual Leakage in Multimodal Safety},
  author={Hu, Xuhao and Liu, Dongrui and Li, Hao and Huang, Xuanjing and Shao, Jing},
  journal={arXiv preprint arXiv:2411.19939},
  year={2024}
}

@article{qwen2.5,
  title={Qwen2. 5 Technical Report},
  author={Yang, An and Yang, Baosong and Zhang, Beichen and Hui, Binyuan and Zheng, Bo and Yu, Bowen and Li, Chengyuan and Liu, Dayiheng and Huang, Fei and Wei, Haoran and others},
  journal={arXiv preprint arXiv:2412.15115},
  year={2024}
}

@misc{gpt4,
      title={GPT-4 Technical Report}, 
      author={OpenAI},
      year={2024},
      eprint={2303.08774},
      archivePrefix={arXiv},
      primaryClass={cs.CL},
      url={https://arxiv.org/abs/2303.08774}, 
}

@article{deepseek,
  title={Deepseek-v2: A strong, economical, and efficient mixture-of-experts language model},
  author={Liu, Aixin and Feng, Bei and Wang, Bin and Wang, Bingxuan and Liu, Bo and Zhao, Chenggang and Dengr, Chengqi and Ruan, Chong and Dai, Damai and Guo, Daya and others},
  journal={arXiv preprint arXiv:2405.04434},
  year={2024}
}

@article{gsm8k,
  title={Training Verifiers to Solve Math Word Problems},
  author={Cobbe, Karl and Kosaraju, Vineet and Bavarian, Mohammad and Chen, Mark and Jun, Heewoo and Kaiser, Lukasz and Plappert, Matthias and Tworek, Jerry and Hilton, Jacob and Nakano, Reiichiro and Hesse, Christopher and Schulman, John},
  journal={arXiv preprint arXiv:2110.14168},
  year={2021}
}

@article{human_eval,
  title={Evaluating large language models trained on code},
  author={Chen, Mark and Tworek, Jerry and Jun, Heewoo and Yuan, Qiming and Pinto, Henrique Ponde De Oliveira and Kaplan, Jared and Edwards, Harri and Burda, Yuri and Joseph, Nicholas and Brockman, Greg and others},
  journal={arXiv preprint arXiv:2107.03374},
  year={2021}
}

@inproceedings{alignment_tax,
author = {Ouyang, Long and Wu, Jeff and Jiang, Xu and Almeida, Diogo and Wainwright, Carroll L. and Mishkin, Pamela and Zhang, Chong and Agarwal, Sandhini and Slama, Katarina and Ray, Alex and Schulman, John and Hilton, Jacob and Kelton, Fraser and Miller, Luke and Simens, Maddie and Askell, Amanda and Welinder, Peter and Christiano, Paul and Leike, Jan and Lowe, Ryan},
title = {Training language models to follow instructions with human feedback},
year = {2022},

booktitle = {Proceedings of the 36th International Conference on Neural Information Processing Systems},

}

@misc{harmbench,
      title={HarmBench: A Standardized Evaluation Framework for Automated Red Teaming and Robust Refusal}, 
      author={Mantas Mazeika and Long Phan and Xuwang Yin and Andy Zou and Zifan Wang and Norman Mu and Elham Sakhaee and Nathaniel Li and Steven Basart and Bo Li and David Forsyth and Dan Hendrycks},
      year={2024},
      eprint={2402.04249},
      archivePrefix={arXiv},
      primaryClass={cs.LG},
      url={https://arxiv.org/abs/2402.04249}, 
}

@inproceedings{circuit_breaker,
  title={Improving alignment and robustness with circuit breakers},
  author={Zou, Andy and Phan, Long and Wang, Justin and Duenas, Derek and Lin, Maxwell and Andriushchenko, Maksym and Kolter, J Zico and Fredrikson, Matt and Hendrycks, Dan},
  booktitle={The Thirty-eighth Annual Conference on Neural Information Processing Systems},
  year={2024}
}

@article{mmlu,
  title={Measuring massive multitask language understanding},
  author={Hendrycks, Dan and Burns, Collin and Basart, Steven and Zou, Andy and Mazeika, Mantas and Song, Dawn and Steinhardt, Jacob},
  journal={arXiv preprint arXiv:2009.03300},
  year={2020}
}

@article{wang2024surgical,
  title={Surgical, Cheap, and Flexible: Mitigating False Refusal in Language Models via Single Vector Ablation},
  author={Wang, Xinpeng and Hu, Chengzhi and R{\"o}ttger, Paul and Plank, Barbara},
  journal={arXiv preprint arXiv:2410.03415},
  year={2024}
}

@inproceedings{panda2024llm,
  title={LLM improvement for jailbreak defense: Analysis through the lens of over-refusal},
  author={Panda, Swetasudha and Nizar, Naveen Jafer and Wick, Michael L},
  booktitle={Neurips Safe Generative AI Workshop 2024},
  year={2024}
}

@article{xstest,
  title={Xstest: A test suite for identifying exaggerated safety behaviours in large language models},
  author={R{\"o}ttger, Paul and Kirk, Hannah Rose and Vidgen, Bertie and Attanasio, Giuseppe and Bianchi, Federico and Hovy, Dirk},
  journal={arXiv preprint arXiv:2308.01263},
  year={2023}
}

@article{orbench,
  title={OR-Bench: An Over-Refusal Benchmark for Large Language Models},
  author={Cui, Justin and Chiang, Wei-Lin and Stoica, Ion and Hsieh, Cho-Jui},
  journal={arXiv preprint arXiv:2405.20947},
  year={2024}
}

@article{oktest,
  title={Navigating the overkill in large language models},
  author={Shi, Chenyu and Wang, Xiao and Ge, Qiming and Gao, Songyang and Yang, Xianjun and Gui, Tao and Zhang, Qi and Huang, Xuanjing and Zhao, Xun and Lin, Dahua},
  journal={arXiv preprint arXiv:2401.17633},
  year={2024}
}

@article{phtest,
  title={Automatic pseudo-harmful prompt generation for evaluating false refusals in large language models},
  author={An, Bang and Zhu, Sicheng and Zhang, Ruiyi and Panaitescu-Liess, Michael-Andrei and Xu, Yuancheng and Huang, Furong},
  journal={arXiv preprint arXiv:2409.00598},
  year={2024}
}

@article{bianchi2023safety,
  title={Safety-tuned llamas: Lessons from improving the safety of large language models that follow instructions},
  author={Bianchi, Federico and Suzgun, Mirac and Attanasio, Giuseppe and R{\"o}ttger, Paul and Jurafsky, Dan and Hashimoto, Tatsunori and Zou, James},
  journal={arXiv preprint arXiv:2309.07875},
  year={2023}
}

@article{varshney2023art,
  title={The art of defending: A systematic evaluation and analysis of llm defense strategies on safety and over-defensiveness},
  author={Varshney, Neeraj and Dolin, Pavel and Seth, Agastya and Baral, Chitta},
  journal={arXiv preprint arXiv:2401.00287},
  year={2023}
}

@article{zeng2024autodefense,
  title={Autodefense: Multi-agent llm defense against jailbreak attacks},
  author={Zeng, Yifan and Wu, Yiran and Zhang, Xiao and Wang, Huazheng and Wu, Qingyun},
  journal={arXiv preprint arXiv:2403.04783},
  year={2024}
}

@article{cao2024nothing,
  title={Nothing in Excess: Mitigating the Exaggerated Safety for LLMs via Safety-Conscious Activation Steering},
  author={Cao, Zouying and Yang, Yifei and Zhao, Hai},
  journal={arXiv preprint arXiv:2408.11491},
  year={2024}
}

@article{arditi2024refusal,
  title={Refusal in Language Models Is Mediated by a Single Direction},
  author={Arditi, Andy and Obeso, Oscar and Syed, Aaquib and Paleka, Daniel and Rimsky, Nina and Gurnee, Wes and Nanda, Neel},
  journal={arXiv preprint arXiv:2406.11717},
  year={2024}
}

@article{zou2023representation,
  title={Representation engineering: A top-down approach to ai transparency},
  author={Zou, Andy and Phan, Long and Chen, Sarah and Campbell, James and Guo, Phillip and Ren, Richard and Pan, Alexander and Yin, Xuwang and Mazeika, Mantas and Dombrowski, Ann-Kathrin and others},
  journal={arXiv preprint arXiv:2310.01405},
  year={2023}
}

@article{zhao2024towards,
  title={Towards comprehensive and efficient post safety alignment of large language models via safety patching},
  author={Zhao, Weixiang and Hu, Yulin and Li, Zhuojun and Deng, Yang and Zhao, Yanyan and Qin, Bing and Chua, Tat-Seng},
  journal={arXiv preprint arXiv:2405.13820},
  year={2024}
}

@article{zou2023GCG,
  title={Universal and transferable adversarial attacks on aligned language models},
  author={Zou, Andy and Wang, Zifan and Carlini, Nicholas and Nasr, Milad and Kolter, J Zico and Fredrikson, Matt},
  journal={arXiv preprint arXiv:2307.15043},
  year={2023}
}

@article{chao2023PAIR,
  title={Jailbreaking black box large language models in twenty queries},
  author={Chao, Patrick and Robey, Alexander and Dobriban, Edgar and Hassani, Hamed and Pappas, George J and Wong, Eric},
  journal={arXiv preprint arXiv:2310.08419},
  year={2023}
}

@inproceedings{zeng2024PAP,
  title={How johnny can persuade llms to jailbreak them: Rethinking persuasion to challenge ai safety by humanizing llms},
  author={Zeng, Yi and Lin, Hongpeng and Zhang, Jingwen and Yang, Diyi and Jia, Ruoxi and Shi, Weiyan},
  booktitle={Proceedings of the 62nd Annual Meeting of the Association for Computational Linguistics (Volume 1: Long Papers)},
  pages={14322--14350},
  year={2024}
}

@article{scaleAI_multi,
  title={Llm defenses are not robust to multi-turn human jailbreaks yet},
  author={Li, Nathaniel and Han, Ziwen and Steneker, Ian and Primack, Willow and Goodside, Riley and Zhang, Hugh and Wang, Zifan and Menghini, Cristina and Yue, Summer},
  journal={arXiv preprint arXiv:2408.15221},
  year={2024}
}

@article{crescendo,
  title={Great, now write an article about that: The crescendo multi-turn llm jailbreak attack},
  author={Russinovich, Mark and Salem, Ahmed and Eldan, Ronen},
  journal={arXiv preprint arXiv:2404.01833},
  year={2024}
}

@article{bju_multi,
  title={Speak Out of Turn: Safety Vulnerability of Large Language Models in Multi-turn Dialogue},
  author={Zhou, Zhenhong and Xiang, Jiuyang and Chen, Haopeng and Liu, Quan and Li, Zherui and Su, Sen},
  journal={arXiv preprint arXiv:2402.17262},
  year={2024}
}

@article{coa,
  title={Chain of Attack: a Semantic-Driven Contextual Multi-Turn attacker for LLM},
  author={Yang, Xikang and Tang, Xuehai and Hu, Songlin and Han, Jizhong},
  journal={arXiv preprint arXiv:2405.05610},
  year={2024}
}

@article{cosafe_multi,
  title={CoSafe: Evaluating Large Language Model Safety in Multi-Turn Dialogue Coreference},
  author={Yu, Erxin and Li, Jing and Liao, Ming and Wang, Siqi and Gao, Zuchen and Mi, Fei and Hong, Lanqing},
  journal={arXiv preprint arXiv:2406.17626},
  year={2024}
}

@article{imposterAI_multi,
  title={Imposter. AI: Adversarial Attacks with Hidden Intentions towards Aligned Large Language Models},
  author={Liu, Xiao and Li, Liangzhi and Xiang, Tong and Ye, Fuying and Wei, Lu and Li, Wangyue and Garcia, Noa},
  journal={arXiv preprint arXiv:2407.15399},
  year={2024}
}

@article{actor_attack,
  title={Derail Yourself: Multi-turn LLM Jailbreak Attack through Self-discovered Clues},
  author={Ren, Qibing and Li, Hao and Liu, Dongrui and Xie, Zhanxu and Lu, Xiaoya and Qiao, Yu and Sha, Lei and Yan, Junchi and Ma, Lizhuang and Shao, Jing},
  journal={arXiv preprint arXiv:2410.10700},
  year={2024}
}

@article{red_queen,
  title={RED QUEEN: Safeguarding Large Language Models against Concealed Multi-Turn Jailbreaking},
  author={Jiang, Yifan and Aggarwal, Kriti and Laud, Tanmay and Munir, Kashif and Pujara, Jay and Mukherjee, Subhabrata},
  journal={arXiv preprint arXiv:2409.17458},
  year={2024}
}

@inproceedings{backdoor_multi,
  title={Securing Multi-turn Conversational Language Models From Distributed Backdoor Attacks},
  author={Tong, Terry and Liu, Qin and Xu, Jiashu and Chen, Muhao},
  booktitle={Findings of the Association for Computational Linguistics: EMNLP 2024},
  pages={12833--12846},
  year={2024}
}

@article{xie2023defending,
  title={Defending chatgpt against jailbreak attack via self-reminders},
  author={Xie, Yueqi and Yi, Jingwei and Shao, Jiawei and Curl, Justin and Lyu, Lingjuan and Chen, Qifeng and Xie, Xing and Wu, Fangzhao},
  journal={Nature Machine Intelligence},
  volume={5},
  number={12},
  pages={1486--1496},
  year={2023},
  publisher={Nature Publishing Group UK London}
}

@article{zheng2024prompt,
  title={Prompt-driven llm safeguarding via directed representation optimization},
  author={Zheng, Chujie and Yin, Fan and Zhou, Hao and Meng, Fandong and Zhou, Jie and Chang, Kai-Wei and Huang, Minlie and Peng, Nanyun},
  journal={arXiv preprint arXiv:2401.18018},
  year={2024}
}

@article{robey2023smoothllm,
  title={Smoothllm: Defending large language models against jailbreaking attacks},
  author={Robey, Alexander and Wong, Eric and Hassani, Hamed and Pappas, George J},
  journal={arXiv preprint arXiv:2310.03684},
  year={2023}
}

@article{cao2023defending,
  title={Defending against alignment-breaking attacks via robustly aligned llm},
  author={Cao, Bochuan and Cao, Yuanpu and Lin, Lu and Chen, Jinghui},
  journal={arXiv preprint arXiv:2309.14348},
  year={2023}
}

@article{liu2024protecting,
  title={Protecting your llms with information bottleneck},
  author={Liu, Zichuan and Wang, Zefan and Xu, Linjie and Wang, Jinyu and Song, Lei and Wang, Tianchun and Chen, Chunlin and Cheng, Wei and Bian, Jiang},
  journal={arXiv preprint arXiv:2404.13968},
  year={2024}
}

@article{jain2023baseline,
  title={Baseline defenses for adversarial attacks against aligned language models},
  author={Jain, Neel and Schwarzschild, Avi and Wen, Yuxin and Somepalli, Gowthami and Kirchenbauer, John and Chiang, Ping-yeh and Goldblum, Micah and Saha, Aniruddha and Geiping, Jonas and Goldstein, Tom},
  journal={arXiv preprint arXiv:2309.00614},
  year={2023}
}

@article{hu2024gradient,
  title={Gradient cuff: Detecting jailbreak attacks on large language models by exploring refusal loss landscapes},
  author={Hu, Xiaomeng and Chen, Pin-Yu and Ho, Tsung-Yi},
  journal={arXiv preprint arXiv:2403.00867},
  year={2024}
}

@article{zhang2023defending,
  title={Defending large language models against jailbreaking attacks through goal prioritization},
  author={Zhang, Zhexin and Yang, Junxiao and Ke, Pei and Huang, Minlie},
  journal={arXiv preprint arXiv:2311.09096},
  year={2023}
}

@article{lu2024sofa,
  title={SoFA: Shielded On-the-fly Alignment via Priority Rule Following},
  author={Lu, Xinyu and Yu, Bowen and Lu, Yaojie and Lin, Hongyu and Yu, Haiyang and Sun, Le and Han, Xianpei and Li, Yongbin},
  journal={arXiv preprint arXiv:2402.17358},
  year={2024}
}

@article{wallace2024instruction,
  title={The instruction hierarchy: Training llms to prioritize privileged instructions},
  author={Wallace, Eric and Xiao, Kai and Leike, Reimar and Weng, Lilian and Heidecke, Johannes and Beutel, Alex},
  journal={arXiv preprint arXiv:2404.13208},
  year={2024}
}

@article{decoupled_sft,
  title={Refuse whenever you feel unsafe: Improving safety in llms via decoupled refusal training},
  author={Yuan, Youliang and Jiao, Wenxiang and Wang, Wenxuan and Huang, Jen-tse and Xu, Jiahao and Liang, Tian and He, Pinjia and Tu, Zhaopeng},
  journal={arXiv preprint arXiv:2407.09121},
  year={2024}
}

@article{safe_unlearning,
  title={Safe unlearning: A surprisingly effective and generalizable solution to defend against jailbreak attacks},
  author={Zhang, Zhexin and Yang, Junxiao and Ke, Pei and Cui, Shiyao and Zheng, Chujie and Wang, Hongning and Huang, Minlie},
  journal={arXiv preprint arXiv:2407.02855},
  year={2024}
}

@article{li2024wmdp,
  title={The wmdp benchmark: Measuring and reducing malicious use with unlearning},
  author={Li, Nathaniel and Pan, Alexander and Gopal, Anjali and Yue, Summer and Berrios, Daniel and Gatti, Alice and Li, Justin D and Dombrowski, Ann-Kathrin and Goel, Shashwat and Phan, Long and others},
  journal={arXiv preprint arXiv:2403.03218},
  year={2024}
}

@article{dubey2024llama,
  title={The llama 3 herd of models},
  author={Dubey, Abhimanyu and Jauhri, Abhinav and Pandey, Abhinav and Kadian, Abhishek and Al-Dahle, Ahmad and Letman, Aiesha and Mathur, Akhil and Schelten, Alan and Yang, Amy and Fan, Angela and others},
  journal={arXiv preprint arXiv:2407.21783},
  year={2024}
}

@article{inan2023llama,
  title={Llama guard: Llm-based input-output safeguard for human-ai conversations},
  author={Inan, Hakan and Upasani, Kartikeya and Chi, Jianfeng and Rungta, Rashi and Iyer, Krithika and Mao, Yuning and Tontchev, Michael and Hu, Qing and Fuller, Brian and Testuggine, Davide and others},
  journal={arXiv preprint arXiv:2312.06674},
  year={2023}
}

@article{ouyang2022training,
  title={Training language models to follow instructions with human feedback},
  author={Ouyang, Long and Wu, Jeffrey and Jiang, Xu and Almeida, Diogo and Wainwright, Carroll and Mishkin, Pamela and Zhang, Chong and Agarwal, Sandhini and Slama, Katarina and Ray, Alex and others},
  journal={Advances in neural information processing systems},
  volume={35},
  pages={27730--27744},
  year={2022}
}

@article{bai2022training,
  title={Training a helpful and harmless assistant with reinforcement learning from human feedback},
  author={Bai, Yuntao and Jones, Andy and Ndousse, Kamal and Askell, Amanda and Chen, Anna and DasSarma, Nova and Drain, Dawn and Fort, Stanislav and Ganguli, Deep and Henighan, Tom and others},
  journal={arXiv preprint arXiv:2204.05862},
  year={2022}
}

@article{chuang2021measuring,
  title={Measuring generalization with optimal transport},
  author={Chuang, Ching-Yao and Mroueh, Youssef and Greenewald, Kristjan and Torralba, Antonio and Jegelka, Stefanie},
  journal={Advances in neural information processing systems},
  volume={34},
  pages={8294--8306},
  year={2021}
}

@article{solomon2020k,
  title={$ k $-Variance: A Clustered Notion of Variance},
  author={Solomon, Justin and Greenewald, Kristjan and Nagaraja, Haikady N},
  journal={arXiv preprint arXiv:2012.06958},
  year={2020}
}

@article{weed2017sharp,
  title={Sharp asymptotic and finite-sample rates of convergence of empirical measures in Wasserstein distance},
  author={Weed, Jonathan and Bach, Francis},
  journal={Bernoulli},
  volume={25},
  number={4A},
  pages={2620--2648},
  year={2019},
  publisher={Bernoulli Society for Mathematical Statistics and Probability}
}

@article{defense_survey,
  title={Jailbreak attacks and defenses against large language models: A survey},
  author={Yi, Sibo and Liu, Yule and Sun, Zhen and Cong, Tianshuo and He, Xinlei and Song, Jiaxing and Xu, Ke and Li, Qi},
  journal={arXiv preprint arXiv:2407.04295},
  year={2024}
}

@article{autodan,
  title={Autodan: Generating stealthy jailbreak prompts on aligned large language models},
  author={Liu, Xiaogeng and Xu, Nan and Chen, Muhao and Xiao, Chaowei},
  journal={arXiv preprint arXiv:2310.04451},
  year={2023}
}

@article{kang2023spliting,
    title={Exploiting programmatic behavior of {LLMs}: Dual-use through standard security attacks},
    author={Kang, Daniel and Li, Xuechen and Stoica, Ion and Guestrin, Carlos and Zaharia, Matei and Hashimoto, Tatsunori},
    journal={arXiv preprint arXiv:2302.05733},
    year={2023}
}

@article{yong2023translation,
    title={Low-resource languages jailbreak {GPT-4}},
    author={Yong, Zheng-Xin and Menghini, Cristina and Bach, Stephen H},
    journal={arXiv preprint arXiv:2310.02446},
    year={2023}
}

@article{zhang2024obfuscation,
  title={Wordgame: Efficient \& effective llm jailbreak via simultaneous obfuscation in query and response},
  author={Zhang, Tianrong and Cao, Bochuan and Cao, Yuanpu and Lin, Lu and Mitra, Prasenjit and Chen, Jinghui},
  journal={arXiv preprint arXiv:2405.14023},
  year={2024}
}

@article{ultrachat,
  title={Enhancing chat language models by scaling high-quality instructional conversations},
  author={Ding, Ning and Chen, Yulin and Xu, Bokai and Qin, Yujia and Zheng, Zhi and Hu, Shengding and Liu, Zhiyuan and Sun, Maosong and Zhou, Bowen},
  journal={arXiv preprint arXiv:2305.14233},
  year={2023}
}

@article{eraser,
  title={Eraser: Jailbreaking defense in large language models via unlearning harmful knowledge},
  author={Lu, Weikai and Zeng, Ziqian and Wang, Jianwei and Lu, Zhengdong and Chen, Zelin and Zhuang, Huiping and Chen, Cen},
  journal={arXiv preprint arXiv:2404.05880},
  year={2024}
}

@article{adaptive_loss,
  title={Adaptive loss weighting for machine learning interatomic potentials},
  author={Ocampo, Daniel and Posso, Daniela and Namakian, Reza and Gao, Wei},
  journal={Computational Materials Science},
  volume={244},
  pages={113155},
  year={2024},
  publisher={Elsevier}
}

@article{safechain,
  title={SafeChain: Safety of Language Models with Long Chain-of-Thought Reasoning Capabilities},
  author={Jiang, Fengqing and Xu, Zhangchen and Li, Yuetai and Niu, Luyao and Xiang, Zhen and Li, Bo and Lin, Bill Yuchen and Poovendran, Radha},
  journal={arXiv preprint arXiv:2502.12025},
  year={2025}
}

@article{R1_assessment,
  title={The Hidden Risks of Large Reasoning Models: A Safety Assessment of R1},
  author={Zhou, Kaiwen and Liu, Chengzhi and Zhao, Xuandong and Jangam, Shreedhar and Srinivasa, Jayanth and Liu, Gaowen and Song, Dawn and Wang, Xin Eric},
  journal={arXiv preprint arXiv:2502.12659},
  year={2025}
}

@article{guo2025deepseek,
  title={Deepseek-r1: Incentivizing reasoning capability in llms via reinforcement learning},
  author={Guo, Daya and Yang, Dejian and Zhang, Haowei and Song, Junxiao and Zhang, Ruoyu and Xu, Runxin and Zhu, Qihao and Ma, Shirong and Wang, Peiyi and Bi, Xiao and others},
  journal={arXiv preprint arXiv:2501.12948},
  year={2025}
}

@article{qian2024towards,
  title={Towards tracing trustworthiness dynamics: Revisiting pre-training period of large language models},
  author={Qian, Chen and Zhang, Jie and Yao, Wei and Liu, Dongrui and Yin, Zhenfei and Qiao, Yu and Liu, Yong and Shao, Jing},
  journal={arXiv preprint arXiv:2402.19465},
  year={2024}
}

@article{zhang2024better,
  title={The better angels of machine personality: How personality relates to llm safety},
  author={Zhang, Jie and Liu, Dongrui and Qian, Chen and Gan, Ziyue and Liu, Yong and Qiao, Yu and Shao, Jing},
  journal={arXiv preprint arXiv:2407.12344},
  year={2024}
}

@article{ren2024identifying,
  title={Identifying semantic induction heads to understand in-context learning},
  author={Ren, Jie and Guo, Qipeng and Yan, Hang and Liu, Dongrui and Zhang, Quanshi and Qiu, Xipeng and Lin, Dahua},
  journal={arXiv preprint arXiv:2402.13055},
  year={2024}
}

@misc{cq,
      title={DEAN: Deactivating the Coupled Neurons to Mitigate Fairness-Privacy Conflicts in Large Language Models}, 
      author={Chen Qian and Dongrui Liu and Jie Zhang and Yong Liu and Jing Shao},
      year={2024},
      eprint={2410.16672},
      archivePrefix={arXiv},
      primaryClass={cs.AI},
      url={https://arxiv.org/abs/2410.16672}, 
}
